%% file: main.tex
\documentclass[11pt]{article}
\usepackage[utf8]{inputenc}
\usepackage[T1]{fontenc}
\usepackage{amsmath, amssymb, amsthm, mathtools}
\usepackage{booktabs}
\usepackage{graphicx}
\usepackage{xcolor}
\usepackage{hyperref}
\usepackage{natbib}
\usepackage{enumitem}
\usepackage{algorithm}
\usepackage{algpseudocode}
\usepackage{caption}
\usepackage{tikz}
\usetikzlibrary{arrows.meta, positioning, shapes.geometric, fit, calc, decorations.pathreplacing}
\usepackage{geometry}
\newtheorem{proposition}{Proposition}
\newtheorem{theorem}{Theorem}

\newtheorem{lemma}{Lemma}

\newcommand{\Rfrac}{x_R}

\title{Delayed Repression and Emergent Instability in Adaptive Multi-Agent Systems}
\author{Igor Itkin\\
\small Independent Researcher, Tel Aviv, Israel\\
\small \texttt{ig.itkin@gmail.com}\\
\small ORCID: \href{https://orcid.org/0009-0004-9513-8463}{0009-0004-9513-8463}}
\date{May 2026}

\begin{document}
\maketitle

\begin{abstract}
Regulatory institutions (from content moderation platforms to financial supervisors) observe, deliberate, and intervene only after a characteristic delay. We ask whether this processing lag alone can destabilize a multi-agent system that would otherwise remain stable, without exogenous shocks, coordination among agents, or malicious actors. We study this in two stages. First, we analyze a delayed replicator equation in which autonomous agents benefit from radical behavior but face punishment based on a lagged institutional alarm signal. We derive a closed-form critical delay beyond which the unique interior equilibrium loses stability through a Hopf bifurcation, and prove via center manifold reduction that the bifurcation is supercritical (bounded oscillations, not explosive growth) for the entire sigmoid response family. Second, we embed $N=240$ agents on a network with reinforcement learning (tabular Q-learning) and cross institutional delay with three decision architectures: fixed-policy, reactive (a memoryless threshold heuristic), and Q-learning. The hierarchy is opposite to the naive expectation that learning amplifies instability. Reactive agents are perfectly stable without delay yet collapse once delay is introduced (96\% runaway by delay$\,{\geq}\,8$); fixed-policy agents are immune (0\% at all delays); Q-learning agents are only partially resilient (66\% at delay$\,{=}\,20$). The destabilizing ingredient is reactivity to delayed signals, not learning: agents that immediately exploit low-alarm windows trigger oscillatory feedback loops, while learning buffers this through punishment memory encoded in value functions. Throughout, ``runaway'' denotes bounded large-amplitude oscillation crossing a radical-fraction threshold, consistent with the supercritical bifurcation, not unbounded growth.
\end{abstract}

\input{sections/01_introduction.tex}
\input{sections/02_model_theory.tex}
\input{sections/03_network_model.tex}
\input{sections/04_experiments.tex}
\input{sections/06_discussion.tex}
\input{sections/07_conclusion.tex}

\section*{Acknowledgments}
The author thanks Ilya Makarov for valuable feedback on the manuscript.

\bibliographystyle{plainnat}
\bibliography{main}

\appendix
\input{sections/appendix_hopf_proof.tex}
\input{sections/appendix_experiment_details.tex}

\end{document}

%% file: sections/01_introduction.tex
\section{Introduction}

Regulatory bodies observe aggregate statistics, deliberate, and intervene only after a characteristic delay. This paper asks what happens in an \emph{adaptive multi-agent system} (a population of agents that modify their behavior in response to observed rewards and punishments) when institutional feedback arrives with such a delay. The answer, at least in our model, is that delay alone can destabilize an otherwise stable system: no exogenous shock is needed. The replicator equation \citep{taylor1978evolutionary, hofbauer1998evolutionary} provides the analytical framework, and network structure shapes the conditions under which autonomy persists \citep{nowak2006five, szabo2007evolutionary, lieberman2005evolutionary}, but the core mechanism we study is temporal: a lag between what agents do and when the institution responds.

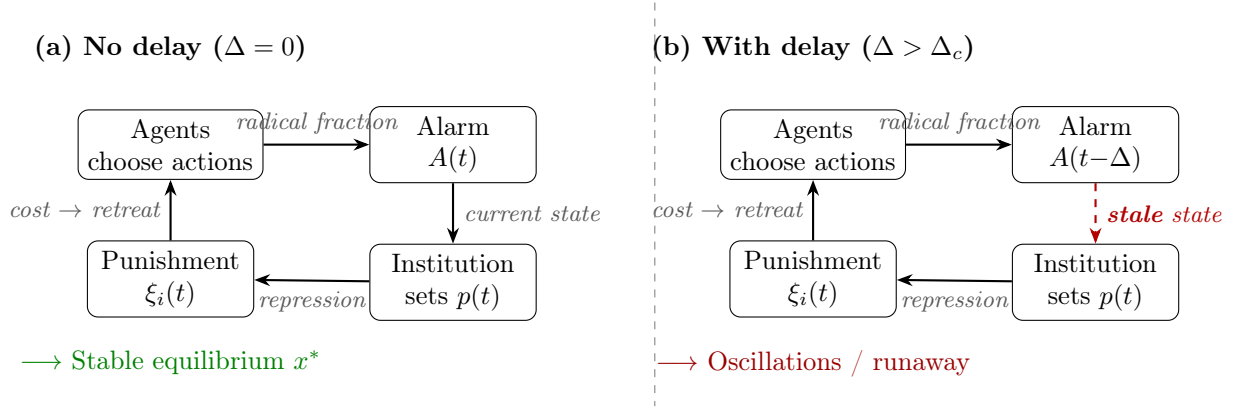
\begin{figure}[t]
\centering
\begin{tikzpicture}[
    node distance=1.2cm and 1.8cm,
    box/.style={draw, rounded corners, minimum width=2.2cm, minimum height=0.7cm, align=center, font=\small},
    arr/.style={-{Stealth[length=6pt]}, thick},
    label/.style={font=\footnotesize\itshape, text=gray!70!black},
]

\begin{scope}[local bounding box=left]
\node[font=\small\bfseries] (titleL) {(a) No delay ($\Delta=0$)};
\node[box, below=0.5cm of titleL] (agentsL) {Agents\\choose actions};
\node[box, right=1.4cm of agentsL] (alarmL) {Alarm\\$A(t)$};
\node[box, below=0.8cm of alarmL] (instL) {Institution\\sets $p(t)$};
\node[box, below=0.8cm of agentsL] (punishL) {Punishment\\$\xi_i(t)$};

\draw[arr] (agentsL) -- node[above, label] {radical fraction} (alarmL);
\draw[arr] (alarmL) -- node[right, label] {current state} (instL);
\draw[arr] (instL) -- node[below, label] {repression} (punishL);
\draw[arr] (punishL) -- node[left, label] {cost $\to$ retreat} (agentsL);

\node[below=0.3cm of punishL, font=\small, text=green!50!black] (resultL) {$\longrightarrow$ Stable equilibrium $x^*$};
\end{scope}

\begin{scope}[xshift=8.5cm, local bounding box=right]
\node[font=\small\bfseries] (titleR) {(b) With delay ($\Delta>\Delta_c$)};
\node[box, below=0.5cm of titleR] (agentsR) {Agents\\choose actions};
\node[box, right=1.4cm of agentsR] (alarmR) {Alarm\\$A(t{-}\Delta)$};
\node[box, below=0.8cm of alarmR] (instR) {Institution\\sets $p(t)$};
\node[box, below=0.8cm of agentsR] (punishR) {Punishment\\$\xi_i(t)$};

\draw[arr] (agentsR) -- node[above, label] {radical fraction} (alarmR);
\draw[arr, red!70!black, dashed] (alarmR) -- node[right, label, text=red!70!black] {\textbf{stale} state} (instR);
\draw[arr] (instR) -- node[below, label] {repression} (punishR);
\draw[arr] (punishR) -- node[left, label] {cost $\to$ retreat} (agentsR);

\node[below=0.3cm of punishR, font=\small, text=red!60!black] (resultR) {$\longrightarrow$ Oscillations / runaway};
\end{scope}

\draw[gray, dashed] ($(left.north east)+(0.6,0.3)$) -- ($(left.south east)+(0.6,-0.3)$);

\end{tikzpicture}
\caption{The delay-destabilization mechanism. (a)~Without delay, the institution observes the current state and adjusts repression in real time; negative feedback drives the system to a stable equilibrium $x^*$. (b)~When institutional processing introduces delay $\Delta$, the alarm signal is stale: agents have already changed behavior by the time repression arrives. Above a critical threshold $\Delta_c$, this stale feedback generates oscillations or runaway. The paper derives $\Delta_c$ analytically (Section~2) and tests which agent architectures are most vulnerable (Section~4).}
\label{fig:mechanism}
\end{figure}

Several lines of work address time delays in replicator dynamics. \citet{kuang1993delay} established the mathematical framework for delay differential equations in population dynamics. In evolutionary game theory, \citet{alboszta2004stability} showed that time delay can destabilize evolutionarily stable strategies in discrete replicator dynamics, while \citet{wesson2016hopf} proved the existence of Hopf bifurcations in two-strategy delayed replicator equations and characterized the critical delay at which limit cycles emerge. The effect is not universal: \citet{iijima2012delayed} found that in a class of discrete delayed dynamics the stability of the mixed equilibrium is preserved independent of the delay, showing that whether delay destabilizes is model-dependent. \citet{mittal2020delayed} analyzed the delayed replicator-mutator equation, showing how delay generates limit cycles even in cooperative regimes. The upshot is that delay can qualitatively change stable equilibria into oscillatory or unstable ones, but the outcome depends on the specific response structure. We build on this line of work because it provides the analytical foundation (the Hopf bifurcation framework for delayed replicator equations) that we extend to nonlinear institutional response functions and validate in an agent-based simulation.

Despite this analytical foundation, the interaction between delayed institutional feedback and adaptive learning agents on structured networks has not been studied. Most analytical results assume well-mixed populations with fixed strategy revision protocols \citep{miekisz2008evolutionary}. Real multi-agent systems, however, involve agents that learn from experience, interact on heterogeneous graphs, and face punishment signals that propagate through institutional channels with variable latency. The literature on evolutionary dynamics on graphs \citep{ohtsuki2006simple, santos2005scalefree, perc2013evolutionary, perc2017statistical, mcavoy2022social} and on multi-agent reinforcement learning \citep{sutton2018reinforcement, zhang2021marl_survey, gronauer2022marl} have proceeded largely separately. Recent work on reward delays in single-agent and multi-agent reinforcement learning (RL) \citep{bouteiller2020delays, zhang2023marl} addresses convergence properties but not the evolutionary stability questions that arise when delayed feedback drives population-level regime transitions. We draw on both traditions because the gap lies at their intersection: delayed institutional feedback (from evolutionary game theory) has not been combined with adaptive learning agents (from multi-agent RL) on structured networks. How does the delay-instability mechanism manifest when agents learn from experience rather than following fixed revision protocols?

Our central finding is counterintuitive: adaptive learning does not amplify delay-induced instability---it partially buffers it. The destabilizing ingredient is not learning but reactivity to delayed signals. Agents that react immediately to low-alarm windows exploit them and trigger oscillatory feedback loops; agents that learn from cumulative punishment history resist this trap. We reach this conclusion through a two-stage approach. First, we analyze the delayed replicator equation with a nonlinear sigmoid response function, specializing the Hopf bifurcation framework of \citet{wesson2016hopf} to derive a closed-form critical delay $\Delta_c$ and prove via center manifold reduction that the bifurcation is supercritical for the entire admissible sigmoid parameter class (Proposition~\ref{prop:supercritical}). This guarantees bounded limit cycles rather than explosive instability above the critical threshold. We validate the analytical result through numerical ordinary differential equation (ODE) integration and a discrete mean-field bridge connecting the continuous theory to the agent-based simulation.

Second, we construct a networked multi-agent simulation with $N=240$ reinforcement learning agents on a modular graph and run a factorial experiment crossing delay with agent decision architecture. Three architectures are compared: non-reactive agents (fixed policy), reactive agents (threshold heuristic without memory), and Q-learning agents (tabular reinforcement learning with cumulative value estimates). The results reveal a clear hierarchy. Reactive agents are the cleanest demonstration of the central claim: with no delay they are fully stable (100\% stable at $\Delta=0$), and delay alone drives them to catastrophic runaway (96\% by delay$\,{\geq}\,8$). Non-reactive (fixed-policy) agents are immune (0\% runaway at all tested values), confirming that the feedback loop is necessary. Q-learning agents achieve partial resilience (66\% at delay${}\,{=}\,20$) by encoding historical punishment into their value functions, though they carry an intrinsic oscillatory tendency at $\Delta=0$ (only ${\approx}42\%$ stable) that the memoryless reactive baseline lacks.

The reduced ODE is a caricature, not a forecast. It identifies a local instability mechanism; the simulations then test whether that mechanism survives contact with a realistic population of adaptive agents on a heterogeneous network. The key empirical question is how different agent architectures interact with delay, and the two-factor design reveals that reactivity, not learning per se, is the destabilizing ingredient.

Section~2 develops the analytical theory (delayed replicator equation, critical delay, Hopf bifurcation proof). Section~3 describes the networked simulation model, including agent learning dynamics and network topologies. Section~4 presents the experiments and results. Section~5 discusses implications and limitations. The companion paper extends the setting to noisy selective control on modular networks. Code and data: \url{https://github.com/YehudaItkin/delayed-repression-instability}.

%% file: sections/02_model_theory.tex
\section{Theory: Delayed Repression Dynamics}

\subsection{Problem formulation}

We consider a population of $N$ agents that choose between conformist and autonomous (radical) behavior over discrete time steps. An institution observes the population state, but with a processing delay of $\Delta$ steps. Based on the delayed observation, the institution sets a repression probability that imposes costs on radical agents. Formally:
\begin{itemize}
\item \textbf{Input:} system parameters $(N, \Delta, k, \text{agent architecture})$, where $k$ is the institutional response sharpness.
\item \textbf{State:} the radical fraction $x(t) \in [0,1]$ at each time step $t$.
\item \textbf{Dynamics:} the delayed replicator equation $\dot{x}(t) = x(t)(1{-}x(t))[a - C\,p(x(t{-}\Delta))]$, where $p(\cdot)$ is the institutional response function and $a = B - S$ is the net autonomy advantage.
\item \textbf{Question:} for what values of $\Delta$ does the unique interior equilibrium $x^*$ lose stability?
\end{itemize}
We formalize this question in a reduced mean-field model (this section), derive the stability boundary $\Delta_c$ in closed form, and then test whether the mechanism survives in a full agent-based simulation where agents maximize cumulative discounted reward $\sum_{t} \gamma^t r_i(t)$ via Q-learning (Section~3).

\subsection{Reduced mean-field model}

The analytical model builds on the replicator equation framework introduced by \citet{taylor1978evolutionary} and developed extensively in \citet{hofbauer1998evolutionary}, incorporating the delay mechanisms analyzed in the general setting by \citet{kuang1993delay} and in evolutionary games specifically by \citet{wesson2016hopf}.

Let $x(t)\in[0,1]$ denote the fraction of agents adopting an autonomous (radical) strategy at time $t$. Non-autonomous agents receive a baseline payoff $S$. Autonomous agents receive benefit $B$, but incur punishment cost $C>0$ with probability determined by a delayed aggregate signal. Define the net autonomy advantage as
\[
    a := B-S.
\]
We study the delayed replicator equation
\begin{equation}
\dot{x}(t)=x(t)(1-x(t))\left[a-Cp(x(t-\Delta))\right],
\label{eq:delayed_rep}
\end{equation}
where $\Delta\ge 0$ is the repression delay and $p\colon[0,1]\to[0,1]$ is a smooth increasing response function representing the probability that the institution activates repression given the observed population state. The logistic growth factor $x(1-x)$ ensures that the dynamics vanish at the population boundaries, as is standard in replicator equations. The key modeling assumption is that the institution observes and responds to the population state at time $t-\Delta$ rather than time $t$, reflecting bureaucratic processing time, information aggregation delays, or deliberation periods.

We model the response function as a sigmoid (logistic function):
\begin{equation}
    p(x)=\sigma(k(x-x_c))=\frac{1}{1+e^{-k(x-x_c)}},
\end{equation}
where $k>0$ controls the sharpness of the institutional response and $x_c\in(0,1)$ is the threshold at which the institution transitions from low to high repression probability. This choice is motivated on three grounds. First, the sigmoid is the standard Fermi update function in stochastic evolutionary game theory \citep{traulsen2006stochastic}, where it parameterizes the sensitivity of strategy revision to payoff differences; here it plays an analogous role for institutional sensitivity to the observed radical fraction. Second, threshold-based responses are empirically characteristic of institutions that tolerate low-level deviation but activate enforcement above a critical level, a pattern common to regulatory agencies, content moderation systems, and immune responses \citep{scheffer2009critical}. Third, given an inflection point $x_c$, the sigmoid is the unique smooth monotone solution of $p'(x)=kp(x)(1-p(x))$ with $p(x_c)=1/2$, making it the natural one-parameter family indexed by sharpness $k$ for any system with logistic-type threshold activation. Larger $k$ corresponds to a more decisive institution that switches abruptly between tolerance and punishment, while smaller $k$ represents a gradual, proportional response.

\subsection{Interior equilibrium and stability without delay}

\begin{proposition}[Interior equilibrium]
The boundary points $x=0$ and $x=1$ are stationary. An interior stationary point $x^*\in(0,1)$ satisfies
\begin{equation}
    p(x^*)=\frac{a}{C}.
\end{equation}
If $p$ is strictly increasing on $[0,1]$, such a point exists and is unique if and only if
\begin{equation}
    p(0)<\frac{a}{C}<p(1).
\end{equation}
For a steep sigmoid with $p(0)\approx0$ and $p(1)\approx1$, this reduces approximately to $0<a<C$.
\end{proposition}

\begin{proof}
At a stationary point of \eqref{eq:delayed_rep}, either $x=0$, $x=1$, or $a-Cp(x^*)=0$. The latter condition is equivalent to $p(x^*)=a/C$. Strict monotonicity of $p$ gives existence and uniqueness precisely when $a/C$ lies in the image of $p$ restricted to the open interval $(0,1)$.
\end{proof}

At the interior equilibrium, institutional punishment exactly offsets the net benefit of autonomy. For the sigmoid response function, the equilibrium is given explicitly by
\begin{equation}
    x^*=x_c+\frac{1}{k}\log\left(\frac{a}{C-a}\right),
\label{eq:xstar}
\end{equation}
provided this value lies in $(0,1)$.

Having located the interior equilibrium, we next establish that it is locally asymptotically stable whenever repression is instantaneous. Delay is therefore necessary for the instabilities studied in subsequent sections.

\begin{theorem}[Local stability without delay]
\label{thm:nodelay}
Assume an interior equilibrium $x^*\in(0,1)$ exists and $p'(x^*)>0$. For $\Delta=0$, $x^*$ is locally asymptotically stable.
\end{theorem}

\begin{proof}
Let
\[
F(x,y)=x(1-x)[a-Cp(y)].
\]
With no delay, the system is $\dot{x}=F(x,x)$. Linearize around $x^*$ by writing $x(t)=x^*+u(t)$. The linear coefficient is
\[
A=\partial_xF(x^*,x^*)+\partial_yF(x^*,x^*).
\]
At the interior equilibrium, $a-Cp(x^*)=0$, so $\partial_xF(x^*,x^*)=(1-2x^*)[a-Cp(x^*)]=0$. Hence
\[
A=-Cx^*(1-x^*)p'(x^*)<0,
\]
which implies local asymptotic stability since the linearized equation $\dot{u}=Au$ decays exponentially.
\end{proof}

\subsection{Delay-induced instability}

When $\Delta>0$, the linearization around $x^*$ takes the form of a scalar delay differential equation. Since $\partial_xF(x^*,x^*)=0$ as shown above, the linearized dynamics involve only the delayed term.

\begin{equation}
    \dot{u}(t)=\beta\, u(t-\Delta),
\label{eq:linear_delay}
\end{equation}
where
\begin{equation}
    \beta=-Cx^*(1-x^*)p'(x^*)<0.
\end{equation}
This is a standard scalar delay equation of the form studied in \citet{kuang1993delay}. Let $b=-\beta>0$. Then \eqref{eq:linear_delay} becomes $\dot{u}(t)=-b\,u(t-\Delta)$, which is the Hayes equation whose stability boundary is classically known.

\begin{theorem}[Critical delay and Hopf crossing]
\label{thm:critical_delay}
For the scalar linearized delay equation \eqref{eq:linear_delay}, the interior equilibrium is locally asymptotically stable for
\[
0\le b\Delta < \frac{\pi}{2}.
\]
At
\begin{equation}
    \Delta_c=\frac{\pi}{2Cx^*(1-x^*)p'(x^*)},
\label{eq:deltac}
\end{equation}
the characteristic equation has a simple conjugate pair of purely imaginary roots $\lambda=\pm ib$ that cross the imaginary axis with positive speed as $\Delta$ increases through $\Delta_c$. Thus $\Delta_c$ is the first local stability boundary and the reduced nonlinear delay differential equation (DDE) satisfies the local Hopf crossing conditions. The criticality of this bifurcation is established in Proposition~\ref{prop:supercritical}.
\end{theorem}

\begin{proof}
Seek solutions of the characteristic equation by substituting $u(t)=e^{\lambda t}$ into \eqref{eq:linear_delay}, obtaining
\[
\lambda=\beta e^{-\lambda\Delta}.
\]
At the first stability crossing, let $\lambda=i\omega$ with $\omega>0$. Since $\beta=-b$,
\[
i\omega=-b e^{-i\omega\Delta}=-b(\cos(\omega\Delta)-i\sin(\omega\Delta)).
\]
Equating real and imaginary parts gives
\[
0=-b\cos(\omega\Delta),\qquad \omega=b\sin(\omega\Delta).
\]
The first condition requires $\cos(\omega\Delta)=0$, so $\omega\Delta=\pi/2+n\pi$ for non-negative integer $n$. The first crossing corresponds to $n=0$, giving $\omega\Delta=\pi/2$ and hence $\omega=b$ from the second equation. Substituting $b=Cx^*(1-x^*)p'(x^*)$ yields \eqref{eq:deltac}.

It remains to verify the transversality condition and the absence of other imaginary-axis roots. Write the characteristic equation as $\lambda+be^{-\lambda\Delta}=0$. Differentiating with respect to $\Delta$:
\[
\frac{d\lambda}{d\Delta}=\frac{b\lambda e^{-\lambda\Delta}}{1-b\Delta e^{-\lambda\Delta}}.
\]
At $\lambda=ib$, $\Delta=\Delta_c=\pi/(2b)$, we have $e^{-ib\Delta_c}=e^{-i\pi/2}=-i$, so $be^{-\lambda\Delta}=-ib$. Substituting:
\[
\frac{d\lambda}{d\Delta}\bigg|_{\Delta=\Delta_c}=\frac{b^2}{1+i\pi/2},
\]
and therefore $\operatorname{Re}(d\lambda/d\Delta)=b^2/(1+\pi^2/4)>0$. The eigenvalue crosses the imaginary axis with positive speed. To verify no other purely imaginary roots exist at $\Delta_c$, note that if $\lambda=i\omega$ is a root, then taking moduli in $i\omega=-be^{-i\omega\Delta_c}$ gives $|\omega|=b$, so the only purely imaginary roots are $\lambda=\pm ib$. By the Hopf bifurcation theorem for delay differential equations \citep{kuang1993delay}, a periodic orbit bifurcates from the equilibrium at $\Delta=\Delta_c$. The direction and stability of this bifurcation are determined by the first Lyapunov coefficient, computed via center manifold reduction in Proposition~\ref{prop:supercritical}.
\end{proof}

\begin{proposition}[Supercritical Hopf bifurcation]
\label{prop:supercritical}
For the sigmoid response $p(x)=\sigma(k(x-x_c))$ and any admissible parameters such that the interior equilibrium $x^*\in(0,1)$ exists (equivalently, $p(0)<a/C<p(1)$; for steep sigmoids this reduces to $0<a<C$), the Hopf bifurcation at $\Delta=\Delta_c$ is supercritical: the first Lyapunov coefficient satisfies $\operatorname{Re}(c_1(0))<0$, the bifurcating periodic orbits are orbitally stable, and their amplitude grows continuously from zero as $\Delta$ increases through $\Delta_c$.
\end{proposition}

The proof proceeds by center manifold reduction following the Hassard--Kazarinoff--Wan formalism \citep{hassard1981theory} applied to the infinite-dimensional DDE phase space \citep{hale1993introduction}. The key steps are: (i)~expansion of the nonlinearity to cubic order around $x^*$; (ii)~projection onto the center eigenspace via the Hale--Verduyn Lunel bilinear form; (iii)~computation of the center manifold corrections $W_{20}$, $W_{11}$; and (iv)~extraction of the first Lyapunov coefficient $c_1(0)$ from the resulting normal form. For the sigmoid response, the closed-form expression for $\operatorname{Re}(c_1(0))$ factors into a positive prefactor times $(\rho-1)\mathcal{B}$, where $\rho = p(x^*) = a/C \in (0,1)$ is the equilibrium repression probability and $\mathcal{B}$ is shown to be strictly positive for all $\rho\in(0,1)$ and $k>0$ by a positive-definiteness argument on a quadratic form. Since $\rho < 1$, it follows that $\operatorname{Re}(c_1(0))<0$ universally. The full derivation, including symbolic verification and numerical validation over $54{,}000$ parameter combinations, is given in Appendix~\ref{app:hopf_proof}.

For the sigmoid response function, the derivative at equilibrium takes the form
\begin{equation}
    p'(x^*) = k\,\frac{a}{C}\left(1-\frac{a}{C}\right),
\end{equation}
Substituting into \eqref{eq:deltac} yields the explicit critical delay
\begin{equation}
    \Delta_c =
    \frac{\pi}
    {2k\,a\,x^*(1-x^*)\left(1-\frac{a}{C}\right)}.
\label{eq:deltac_sigmoid}
\end{equation}
This expression reveals that the critical delay decreases with increasing sharpness $k$, with increasing autonomy advantage $a$ (provided $a<C$), and is minimized near mixed-population states where $x^*(1-x^*)$ is large.

\subsection{Summary and testable hypotheses}

Prior work established the existence of Hopf bifurcations in delayed replicator equations with symmetric, linear payoffs \citep{wesson2016hopf, wesson2016hopfRPS}. \citet{benkhalifa2018distributed} showed that the shape of the delay distribution (Dirac, uniform, Gamma) affects the Hopf threshold, and \citet{wettergren2023replicator} demonstrated that asymmetric delays on costs versus benefits produce alternating stability--instability windows absent in the single-delay case. In feedback-evolving games, \citet{yan2021cooperator} found that delayed environmental feedback from cooperation, combined with punishment of defectors, generates oscillations through the same Hopf mechanism, a parallel finding where the delay acts on agent contributions rather than institutional observation. More recently, \citet{gao2025bifurcation} combined logistic population growth with strategy-dependent delays, producing richer bifurcation structure than the constant-population case. In multi-agent reinforcement learning, \citet{chen2020delay} formalized delay-aware Markov games and showed that delay from one agent propagates to others' learning, while \citet{derman2021acting} proved that deterministic non-stationary Markov policies suffice under execution delay without state augmentation, which provides theoretical backing for why Q-learning handles delayed environments.

The present analysis extends this body of work in two directions. First, we replace the symmetric payoff structure with an asymmetric institutional response: a nonlinear sigmoid function $p(x)$ that models threshold-based repression (Section~2.2). This changes both the equilibrium condition ($p(x^*)=a/C$ rather than a symmetric Nash condition) and the form of the critical delay (Equation~\ref{eq:deltac_sigmoid}), which now depends on the sharpness parameter $k$ controlling institutional decisiveness. Second, we prove that the Hopf bifurcation is supercritical \emph{universally} across the sigmoid parameter class (Proposition~\ref{prop:supercritical}), not just for isolated parameter values. Prior results left the bifurcation direction as a numerical observation; we provide a closed-form proof via center manifold reduction, validated symbolically and numerically over $54{,}000$ parameter combinations (Appendix~\ref{app:hopf_proof}). Together, these results yield a complete analytical characterization: for any sigmoid response, the delayed system transitions from a stable equilibrium to bounded limit cycles at a critical delay $\Delta_c$ that decreases with sharpness $k$ and is maximally fragile near the institutional threshold.

The analytical results generate three testable hypotheses for the networked agent-based simulation (Section~3). We state them here to motivate the experimental design in Section~4.

\begin{enumerate}
\item \textbf{Delay destabilizes.} Increasing the institutional delay $\Delta$ should move the system from stable equilibrium to oscillatory dynamics and, beyond the local bifurcation, potentially to runaway-crackdown regimes. This follows directly from Theorem~\ref{thm:critical_delay}: the stability margin shrinks as $\Delta$ grows.
\item \textbf{Sharpness amplifies.} Increasing the response sharpness $k$ should reduce the stability margin by lowering $\Delta_c$ (Equation~\ref{eq:deltac_sigmoid}), so that at fixed delay the system crosses from stable to unstable as $k$ increases.
\item \textbf{Mixed populations are most fragile.} Fragility should be greatest where the product $x^*(1-x^*)p'(x^*)$ is large. For a sigmoid response, this maximum occurs near the response midpoint $x_c$ when the population is neither overwhelmingly conformist nor overwhelmingly autonomous.
\end{enumerate}

These hypotheses concern local stability. Theorem~\ref{thm:critical_delay} proves that the equilibrium loses stability via a Hopf bifurcation, and Proposition~\ref{prop:supercritical} establishes that this bifurcation is supercritical: the resulting periodic orbit is orbitally stable with amplitude growing continuously from zero. The theory does not, by itself, characterize global dynamics such as large-amplitude runaway or crackdown. The transition from continuous ODE to discrete networked simulation introduces three additional complexity layers: (i)~finite populations with stochastic action selection, (ii)~heterogeneous network structure, and (iii)~adaptive agents that learn from experience. The experiments in Section~4 test whether the delay-destabilization mechanism survives these extensions, and which agent architecture is most vulnerable.

%% file: sections/03_network_model.tex
\section{Networked Simulation Model}

The mean-field ODE (Section~2) proves that delay can destabilize an equilibrium, but it assumes a well-mixed population of identical agents with a fixed strategy-revision rule. Real multi-agent systems violate all three assumptions: agents are heterogeneous, interact on structured networks, and adapt their behavior through learning. This section constructs a simulation that relaxes these assumptions. The goal is not to replicate the ODE dynamics quantitatively (the integer delay steps in the simulation do not map onto the continuous time units of the ODE) but to test whether the three hypotheses from Section~2.5 survive in a substantially richer system.

\subsection{Agents and actions}

Each agent $i\in\{1,\ldots,N\}$ chooses an action at each discrete time step $t$:
\[
    a_i(t)\in\{L,M,R\},
\]
corresponding to loyal, moderate, and radical behavior. These actions are mapped to a numerical activity level $f(L)=0$, $f(M)=1$, $f(R)=2$. Loyal agents conform fully to institutional expectations, moderate agents exercise limited autonomy, and radical agents pursue maximum autonomy at the risk of attracting punishment. Each agent is characterized by a learning rate $\alpha_i$, an exploration parameter $\epsilon_i$, and a detectability score $d_i$ that modulates how visible the agent's behavior is to the institutional observer.

\subsection{Network topologies}

Agents interact on a directed graph $G=(V,E)$ where edges represent influence relationships. The simulation supports three topology classes. Erd\H{o}s--R\'enyi random graphs provide a homogeneous-mixing baseline with approximately Poisson degree distribution \citep{szabo2007evolutionary}. Scale-free networks generated via preferential attachment produce heavy-tailed degree distributions in which hub nodes exert disproportionate influence \citep{santos2005scalefree}. Stochastic block model graphs create modular community structure with dense intra-community connectivity and sparse inter-community links \citep{holland1983stochastic, girvan2002community}.

In the modular topology, bridge nodes are identified using betweenness centrality \citep{freeman1977betweenness}, which measures the fraction of shortest paths between all node pairs that pass through a given node. Nodes with high betweenness centrality serve as conduits for information and influence between communities. This definition ensures that bridge status reflects actual inter-community connectivity rather than merely high intra-community degree.

\subsection{Aggregate alarm and delayed repression}

The institutional observer computes an aggregate alarm signal by averaging weighted activity levels across the population:
\begin{equation}
    A(t)=\frac{1}{N}\sum_{i=1}^{N} d_i\, f(a_i(t)).
\end{equation}
The detectability weights $d_i$ allow heterogeneous visibility: agents in exposed positions or with high-profile behavior contribute more to the institutional alarm than those operating in obscurity.

The repression probability for agent $i$ at time $t$ depends on the delayed alarm:
\begin{equation}
    p_i(t)=u_t\,\sigma\bigl(k(A(t-\Delta)-A_c)\bigr)\,d_i,
\end{equation}
where $\sigma(\cdot)$ is the logistic sigmoid from Section~2, $\Delta$ is the institutional delay measured in discrete time steps, $k$ is the response sharpness, $A_c$ is the alarm threshold, and $u_t \in [0, 2.5]$ is the regulator force (the institutional control intensity). The regulator force $u_t$ may be fixed at a constant value, adjusted by a heuristic rule, or selected by a learning regulator agent. The product $\sigma(\cdot)\,d_i$ ensures that more detectable agents face higher punishment probability when the system is under repression.

\subsection{Local influence and rewards}

Each agent observes its local neighborhood through an influence signal:
\begin{equation}
    I_i(t)=\sum_{j\in\mathcal{N}^{-}(i)} w_{ji}\, f(a_j(t)),
\end{equation}
where $\mathcal{N}^{-}(i)$ denotes the set of agents with directed edges into $i$ and $w_{ji}$ are influence weights. The agent's immediate reward combines a benefit from its action, a social influence component, and a punishment cost:
\begin{equation}
    r_i(t)=B(a_i(t),\chi_i)+\lambda\, I_i(t)\,f(a_i(t))-\xi_i(t)\,C(a_i(t),\kappa_i),
\end{equation}
where $\xi_i(t)\sim\mathrm{Bernoulli}(p_i(t))$ is a stochastic punishment indicator (the agent either receives the full cost or nothing at each step), $\chi_i$ is a charisma parameter scaling the benefit of radical behavior, $\lambda$ is the influence coupling strength, and $\kappa_i$ scales the punishment cost. The benefit function is $B(a,\chi) = \chi \cdot f(a)$ (linear in activity level, so radical actions yield the highest benefit), and the cost function is $C(a,\kappa) = \kappa \cdot f(a)$ (punishment is proportional to activity level). This incentive gradient toward radicalism is counterbalanced by the repression mechanism.

\subsection{Learning dynamics}

In learning variants, each agent employs tabular Q-learning \citep{sutton2018reinforcement}, a standard reinforcement learning algorithm in which the agent maintains a table of estimated long-run values $Q(s,a)$ for each state--action pair and updates them toward observed rewards. The agent uses a compact state representation. The observation state for agent $i$ at time $t$ is a tuple comprising a discretized local influence bucket, a discretized delayed alarm bucket, a binary recent-punishment indicator, and a binary bridge-membership flag. This yields a state space of manageable size ($3 \times 3 \times 2 \times 2 = 36$ states) that agents can explore within the simulation horizon. The Q-update rule is standard:
\[
Q_i(s,a) \leftarrow Q_i(s,a) + \alpha_i\bigl[r_i(t) + \gamma \max_{a'} Q_i(s',a') - Q_i(s,a)\bigr],
\]
where $s$ is the current observation state, $s'$ is the next observation state (computed from the next time step's delayed alarm and local influence), and $\gamma$ is the discount factor.

Action selection follows an $\epsilon$-greedy policy with exploration rate $\epsilon_i = 0.08$ (a standard choice in tabular RL ensuring sufficient exploration without excessive noise; \citealp{sutton2018reinforcement}, Chapter~2). The learning rate is $\alpha_i = 0.10$ and the discount factor $\gamma = 0.95$, both standard defaults for tabular Q-learning in environments with moderate horizon length \citep{sutton2018reinforcement}. We verified robustness with a $3\times3$ sweep over $\alpha \in \{0.05, 0.1, 0.2\}$ and $\gamma \in \{0.9, 0.95, 0.99\}$: across all nine combinations the Q-learning runaway rate is ${\approx}15\%$ at $\Delta=0$, $35\%$ at $\Delta=8$, and $65$--$75\%$ at $\Delta=20$, so both the monotonic delay--runaway relationship and the architecture hierarchy are preserved (results under \texttt{experiments/results/paper1\_v2/alpha\_gamma\_robust}). The learning model is deliberately minimal: it provides agents with the capacity to detect and exploit temporal patterns in repression without pretending to model human cognition.

In fixed-policy variants, agents do not learn. Instead, they sample actions from a stationary probability distribution over $\{L,M,R\}$ that does not depend on the environment state. This provides a baseline against which to measure the destabilizing potential of adaptive behavior under delayed feedback.

\subsection{Simulation loop}

Algorithm~\ref{alg:sim} summarizes the complete simulation loop, showing how the components defined above connect into a training procedure. Each agent's objective is to maximize cumulative discounted reward $\sum_{t=0}^{T-1} \gamma^t r_i(t)$; there is no explicit loss function to minimize, because the Q-learning update rule (above) approximates the optimal policy through temporal-difference bootstrapping rather than gradient descent.

\begin{figure}[ht]
\centering
\begin{minipage}{0.92\linewidth}
\begin{algorithmic}[1]
\small
\State \textbf{Input:} graph $G$, delay $\Delta$, sharpness $k$, horizon $T$, seeds
\State Initialize Q-tables $Q_i(s,a) \leftarrow 0$ for all agents $i$, states $s$, actions $a$
\State Initialize alarm history buffer $H \leftarrow [0, \ldots, 0]$ of length $\Delta + 1$
\For{$t = 0, 1, \ldots, T-1$}
    \State Compute delayed alarm: $A_{\mathrm{delayed}} \leftarrow H[t - \Delta]$ \hfill\emph{(stale observation)}
    \For{each agent $i$}
        \State Observe state $s_i \leftarrow (\text{influence bucket}, \text{alarm bucket}, \text{punished?}, \text{bridge?})$
        \State Select action $a_i \leftarrow \epsilon\text{-greedy}(Q_i, s_i, \epsilon_i)$
    \EndFor
    \State Compute current alarm: $A(t) \leftarrow \frac{1}{N}\sum_i d_i \cdot f(a_i)$; append to $H$
    \State Compute repression: $p_i(t) \leftarrow u_t \cdot \sigma(k(A_{\mathrm{delayed}} - A_c)) \cdot d_i$
    \State Sample punishment: $\xi_i(t) \sim \mathrm{Bernoulli}(p_i(t))$
    \State Compute rewards: $r_i(t) \leftarrow B(a_i, \chi_i) + \lambda\, I_i(t) \cdot f(a_i) - \xi_i(t) \cdot C(a_i, \kappa_i)$
    \For{each agent $i$ (learning variants only)}
        \State Observe next state $s'_i$ from step $t{+}1$ observations
        \State Update: $Q_i(s_i, a_i) \leftarrow Q_i(s_i, a_i) + \alpha_i[r_i(t) + \gamma \max_{a'} Q_i(s'_i, a') - Q_i(s_i, a_i)]$
    \EndFor
\EndFor
\State \textbf{Output:} time-series history $\{A(t), \Rfrac(t), \text{regime label}\}$
\end{algorithmic}
\end{minipage}
\captionof{algorithm}{Simulation loop for one run. The key feature is that repression at step $t$ is computed from the alarm at step $t{-}\Delta$ (line~5), creating the delayed feedback loop that drives instability.}
\label{alg:sim}
\end{figure}

%% file: sections/04_experiments.tex
\section{Experiments}

\begin{figure}[t]
\centering
\begin{tikzpicture}[
    stage/.style={draw, rounded corners, fill=blue!8, minimum width=3.0cm, minimum height=0.9cm, align=center, font=\small},
    hyp/.style={draw, rounded corners, fill=orange!10, minimum width=2.4cm, minimum height=0.6cm, align=center, font=\footnotesize},
    result/.style={draw, rounded corners, fill=green!10, minimum width=2.4cm, minimum height=0.6cm, align=center, font=\footnotesize},
    arr/.style={-{Stealth[length=5pt]}, thick},
    note/.style={font=\scriptsize\itshape, text=gray!70!black},
]

\node[stage] (ode) {Stage 1:\\Delayed Replicator ODE};
\node[stage, right=2.2cm of ode] (sim) {Stage 2:\\Networked Simulation};
\node[stage, right=2.2cm of sim] (class) {Regime\\Classification};

\node[hyp, below=0.7cm of ode] (h1) {H1: Delay destabilizes};
\node[hyp, below=0.15cm of h1] (h2) {H2: Sharpness amplifies};
\node[hyp, below=0.15cm of h2] (h3) {H3: Mixed pop.\ fragile};

\node[result, below=0.7cm of sim] (e0) {Exp.~1: ODE valid.};
\node[result, below=0.15cm of e0] (e1) {Exp.~2: Discrete bridge};
\node[result, below=0.15cm of e1] (e23) {Exp.~3--4: Sweeps};
\node[result, below=0.15cm of e23] (e4) {Exp.~5: Crossed design};
\node[result, below=0.15cm of e4] (e5) {Exp.~6: RL regulator};

\node[result, below=0.7cm of class] (rq1) {Q1: Survives? \checkmark};
\node[result, below=0.15cm of rq1] (rq2) {Q2: Sharpness? \checkmark};
\node[result, below=0.15cm of rq2] (rq3) {Q3: Learning buffers};
\node[result, below=0.15cm of rq3] (rq4) {Q4: RL regulator?};

\draw[arr] (ode) -- node[above, note] {hypotheses} (sim);
\draw[arr] (sim) -- node[above, note] {50 seeds/cond.} (class);

\draw[arr, gray] (ode) -- (h1);
\draw[arr, gray] (sim) -- (e0);
\draw[arr, gray] (class) -- (rq1);

\node[above=0.15cm of ode, note] {Section 2};
\node[above=0.15cm of sim, note] {Sections 3--4};
\node[above=0.15cm of class, note] {Section 4};

\end{tikzpicture}
\caption{Solution pipeline. Stage~1 derives analytical predictions from the delayed replicator ODE. Stage~2 tests these predictions in a networked agent-based simulation with three decision architectures. Regime classification assigns each run to stable, oscillatory, or runaway. The mapping from hypotheses (H1--H3) through experiments (1--6) to research questions (Q1--Q4) is detailed in Table~\ref{tab:experiment_map}.}
\label{fig:pipeline}
\end{figure}
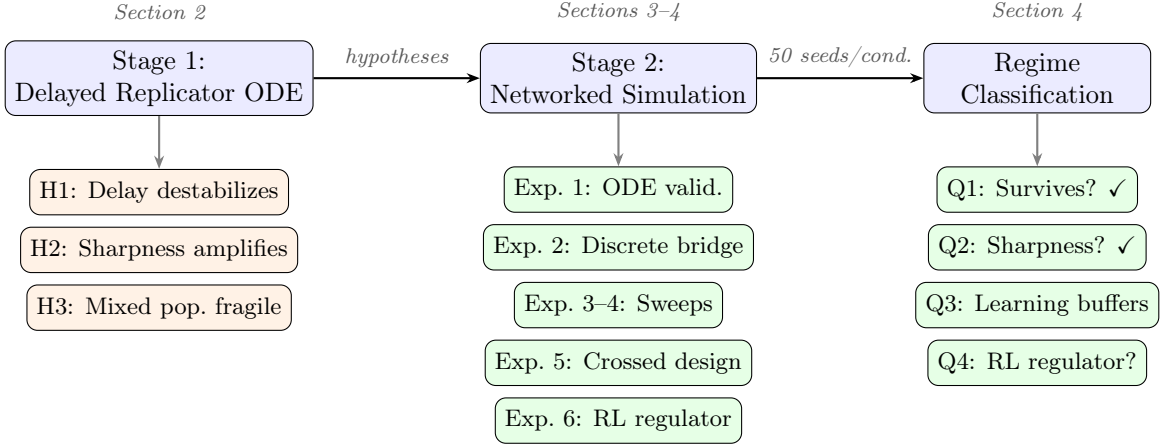

The analytical theory (Section~2) establishes that delay can destabilize a well-mixed population with a fixed strategy-revision rule. The simulation model (Section~3) introduces three layers of realism: finite heterogeneous populations, network structure, and adaptive learning via reinforcement learning. The experiments address the following research questions:

\begin{enumerate}
\item Does the delay-destabilization mechanism survive in the networked simulation? The ODE predicts instability above $\Delta_c$; the simulation adds noise, discretization, and network effects that could either amplify or suppress this mechanism.
\item How does sigmoid sharpness $k$ interact with delay? The theory predicts that higher $k$ lowers $\Delta_c$ (Equation~\ref{eq:deltac_sigmoid}). We test whether the simulation reproduces this monotonic relationship.
\item Does adaptive learning amplify or buffer delay-induced instability? The naive expectation is that learning agents exploit delayed low-alarm windows more effectively, amplifying instability. The alternative is that cumulative learning provides implicit memory that dampens oscillations.
\item Can an adaptive regulator compensate for the destabilizing effect of adaptive agents? If agents learn to exploit delay, can an RL-trained regulator learn to counteract them?
\end{enumerate}

Table~\ref{tab:experiment_map} maps each experiment to the research question it addresses and the key controlled variables.

\begin{table}[ht]
\centering
\begin{tabular}{llll}
\toprule
Experiment & Research question & Varied factor & Controlled factors \\
\midrule
Experiment~1: ODE validation & Baseline & $\Delta/\Delta_c$ & Continuous ODE \\
Experiment~2: Discrete mean-field & Q1 (discretization) & $\eta$, $\Delta$ & Mean-field, no network \\
Experiment~3: Delay sweep & Q1 & $\Delta$ & $k{=}20$, Q-learning \\
Experiment~4: Sharpness sweep & Q2 & $k$ & $\Delta{=}15$, Q-learning \\
Experiment~5: Crossed design & Q3 & $\Delta \times$ architecture & $k{=}10$, 3 architectures \\
Experiment~6: RL regulator & Q4 & Regulator type & $\Delta{=}6$, Q-learning \\
\bottomrule
\end{tabular}
\caption{Mapping of experiments to research questions. Each experiment varies one factor while controlling others to isolate the effect. Experiment~1 and Experiment~2 validate the analytical theory; Experiment~3 through Experiment~6 test the hypotheses in the full simulation.}
\label{tab:experiment_map}
\end{table}

The progression is deliberate: Experiment~1 confirms the ODE theory is mathematically correct; Experiment~2 bridges continuous and discrete dynamics; Experiment~3 and Experiment~4 establish dose-response curves for delay and sharpness; Experiment~5 is the central experiment testing the interaction between delay and agent architecture; Experiment~6 is exploratory. All results reported in this section use 50 independent seeds per condition. Raw time-series histories, summary statistics, and configuration manifests are stored under \texttt{experiments/results/paper1\_v2/} with per-experiment subdirectories. The progression from ODE validation through discrete mean-field to full network simulation reveals both the robustness of the core delay-instability mechanism and the quantitative gaps introduced by agent heterogeneity and adaptive learning.

\subsection{Experimental setup}

\paragraph{Shared settings.} All networked experiments use $N=240$ agents on a modular stochastic block model graph (6 communities of 40 nodes, intra-community edge probability $p_{\mathrm{in}}=0.08$, inter-community $p_{\mathrm{out}}=0.004$). The population size $N{=}240$ is large enough for stable mean-field-like alarm statistics but small enough for complete Q-table exploration within 500 steps. The 6-community modular structure follows standard stochastic block model designs \citep{holland1983stochastic} and produces well-defined bridge nodes for the companion paper's analysis; the intra/inter edge probabilities ($0.08/0.004$) yield a modularity ratio of $p_{\mathrm{in}}/p_{\mathrm{out}}=20$, ensuring clear community separation. The simulation horizon is 500 time steps, sufficient for Q-learning agents to reach stable policies (convergence diagnostic below confirms Q-values equilibrate by step 100--200). We use 50 independent seeds per condition throughout, with results frozen after completion to ensure reproducibility. The sharpness parameter $k$ is set high ($k=10$ or $k=20$) to ensure the system operates above the critical threshold $\Delta_c$ at the tested delay values; lower $k$ would require longer delays to produce measurable effects. The alarm threshold $A_c$ is set to the theoretical equilibrium level. These parameter choices follow from the analytical predictions: we choose settings where the theory predicts instability and test whether the simulation confirms it.

\paragraph{Baselines and comparison strategy.} This paper studies a mechanism (delay-induced instability), not a method that competes against alternatives on a benchmark. The natural baselines are therefore \emph{ablation baselines}: conditions that remove specific components of the mechanism to isolate their contribution. Fixed-policy agents serve as the null model (no feedback loop, so delay cannot destabilize). Reactive agents isolate the effect of memoryless reactivity. Q-learning agents add adaptive memory. The comparison across these three architectures answers Q3. No prior work has studied the specific interaction between institutional delay, agent learning, and network structure that we examine; accordingly, we compare agent architectures rather than competing models from the literature. The ODE validation (Experiment~1) and discrete mean-field (Experiment~2) serve as additional baselines connecting the simulation to the analytical theory.

\paragraph{Regime classification.} Each simulation run is classified into one of three regimes based on tail-half statistics (the final 50\% of the time series, after transient dynamics decay). \emph{Stable}: the radical fraction remains below the runaway threshold $R_{\max}=0.40$ and oscillation amplitude is small. \emph{Oscillatory}: substantial periodic variation in radical fraction (detected by spectral analysis: peak-to-total power ratio $>0.25$ with amplitude $>0.08$, or a fallback test of tail amplitude $>0.12$ or tail standard deviation $>0.04$) without exceeding $R_{\max}$. \emph{Runaway}: the maximum radical fraction exceeds $R_{\max}=0.40$. We emphasize that ``runaway'' denotes crossing this threshold, not unbounded growth: across all runs so labeled, the tail-period mean radical fraction never exceeds $0.30$, consistent with the supercritical-Hopf prediction of bounded limit cycles rather than literal radical domination. The threshold $R_{\max}=0.40$ is chosen as a substantively meaningful level at which radical agents dominate; we verify that the central ordering (fixed $\leq$ Q-learning $\leq$ reactive in delay sensitivity) is preserved across thresholds from 0.30 to 0.50 (Section~5).

\subsection{Theory validation (Experiment~1, Experiment~2)}

\paragraph{Experiment~1: ODE validation.} This experiment validates the analytical theory by numerically integrating the delayed replicator equation~\eqref{eq:delayed_rep} directly, sweeping $\Delta$ from $0$ to $4\Delta_c$. It is the ground-truth baseline: if the ODE does not bifurcate at $\Delta_c$, the theory is wrong and the simulation experiments are moot. The integration uses parameter values $a{=}2.0$, $C{=}7.0$, $k{=}10$, $x_c{=}0.72$, representative magnitudes chosen so that the interior equilibrium $x^*{\approx}0.63$ sits at a non-degenerate point away from the boundaries, with critical delay $\Delta_c{\approx}0.47$. The supercriticality proof (Appendix~\ref{app:hopf_proof}) shows the qualitative bifurcation is insensitive to these specific values across all admissible parameters.

Direct numerical integration of Equation~\eqref{eq:delayed_rep} confirms Theorem~\ref{thm:critical_delay}. Figure~\ref{fig:ode_bifurcation} shows the bifurcation diagram: below $\Delta_c$, trajectories converge to $x^*$; at $\Delta/\Delta_c=1$, stable limit cycles emerge with amplitude growing continuously from zero, consistent with the supercritical Hopf bifurcation. Figure~\ref{fig:ode_delay_sweep} shows representative trajectories, and Figure~\ref{fig:ode_k_sweep} shows that increasing sharpness $k$ at fixed $\Delta/\Delta_c = 1.5$ sharpens the nonlinear response and increases limit-cycle amplitude.

\begin{figure}[ht]
\centering
\includegraphics[width=0.85\textwidth]{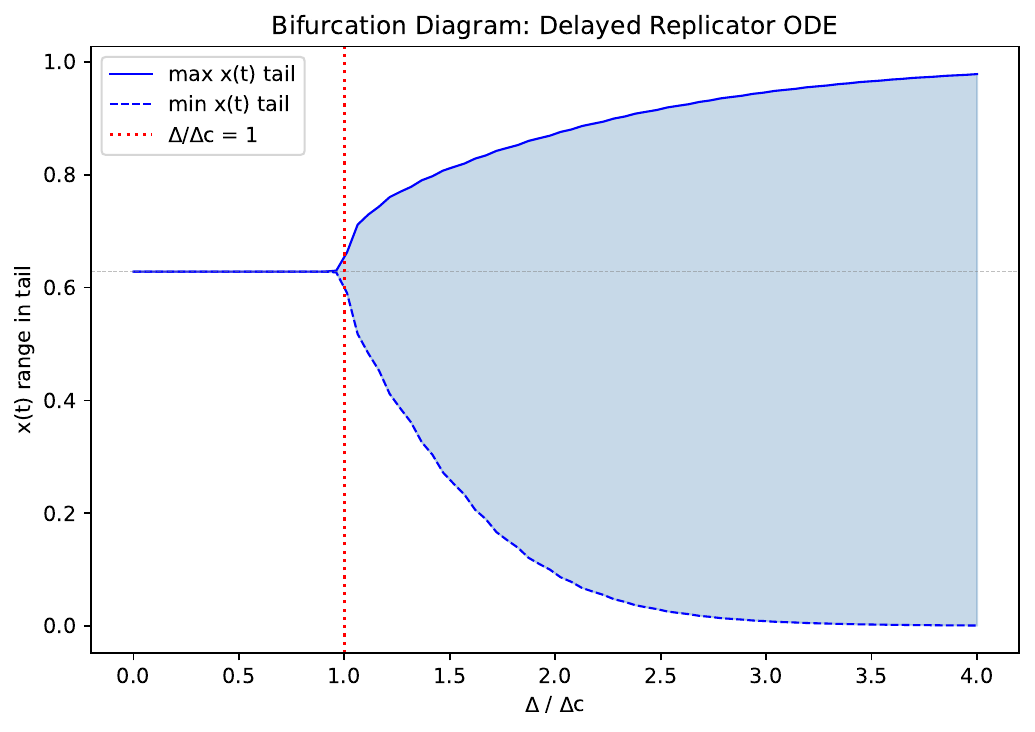}
\caption{Bifurcation diagram for the delayed replicator ODE (Experiment~1). Horizontal axis: normalized delay $\Delta/\Delta_c$. Below $\Delta_c$, the equilibrium is stable. At $\Delta/\Delta_c=1$, a Hopf bifurcation produces stable limit cycles with continuously growing amplitude (solid: numerical envelope; dashed: theoretical threshold). \textbf{Conclusion:} the analytical theory (Theorem~\ref{thm:critical_delay}) is confirmed exactly---the instability occurs precisely at the predicted threshold.}
\label{fig:ode_bifurcation}
\end{figure}

\begin{figure}[ht]
\centering
\includegraphics[width=0.85\textwidth]{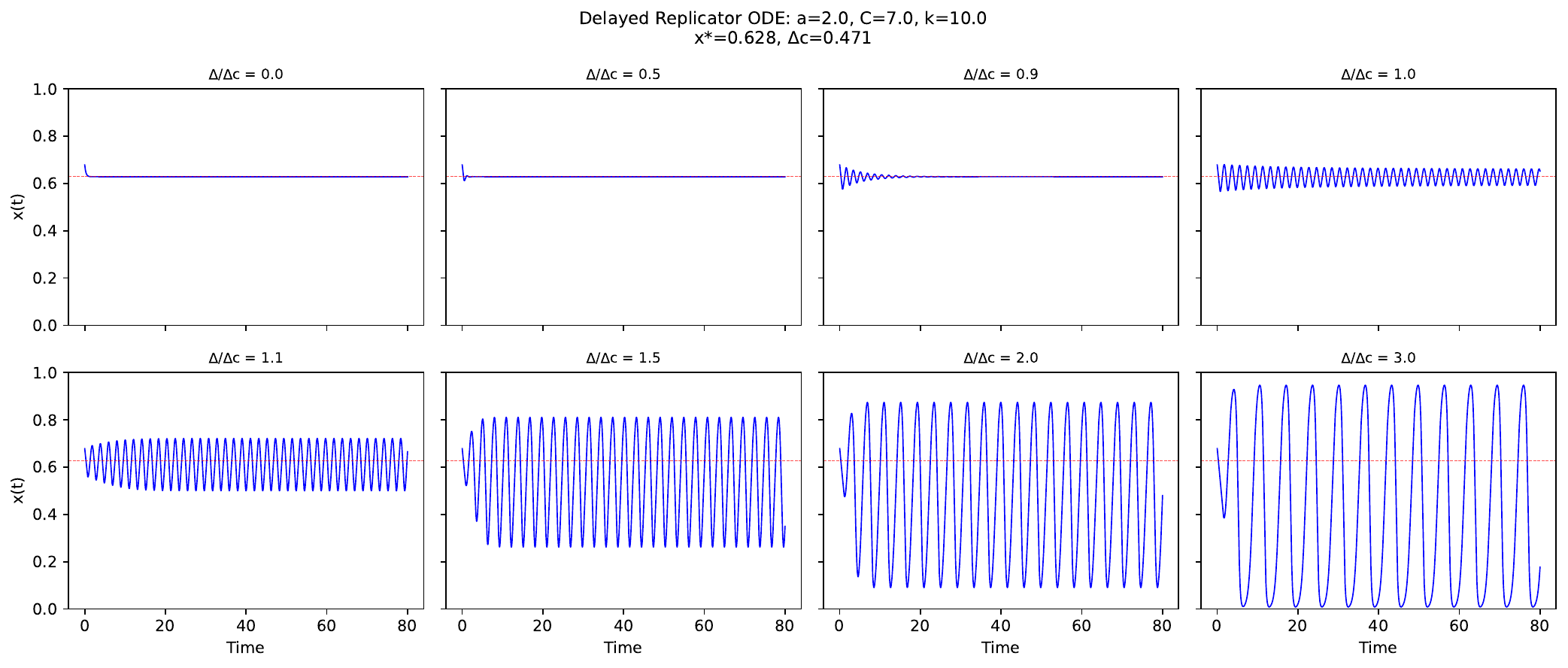}
\caption{Representative ODE trajectories at six delay levels (Experiment~1). Below $\Delta_c$: damped convergence to equilibrium. Above: sustained limit cycles with amplitude growing with $\Delta$. The transition is sharp, consistent with the supercritical Hopf bifurcation (Proposition~\ref{prop:supercritical}).}
\label{fig:ode_delay_sweep}
\end{figure}

\begin{figure}[ht]
\centering
\includegraphics[width=0.85\textwidth]{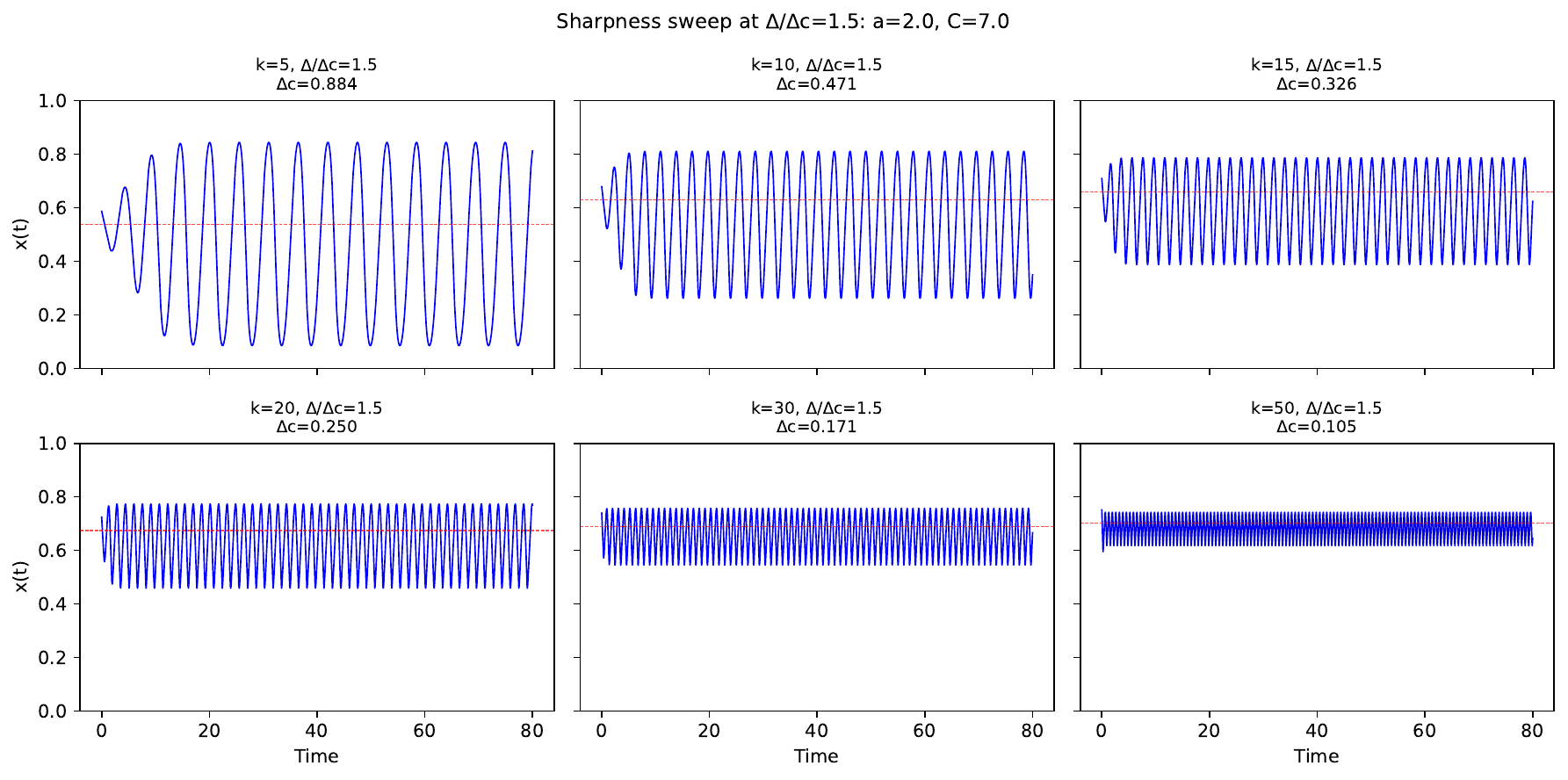}
\caption{Effect of sharpness $k$ on ODE dynamics at fixed $\Delta/\Delta_c = 1.5$ (Experiment~1). All panels are at the same distance beyond the stability boundary, yet higher $k$ produces sharper, larger-amplitude limit cycles saturating near the population boundaries. \textbf{Conclusion:} sharpness amplifies the nonlinear consequences of delay-induced instability, confirming Hypothesis~2.}
\label{fig:ode_k_sweep}
\end{figure}

\paragraph{Experiment~2: Discrete mean-field.} This experiment bridges continuous ODE dynamics and discrete multi-agent simulation. Discretization itself can introduce instabilities absent in the continuous limit \citep{alboszta2004stability}. We study a discrete-time mean-field system with step size $\eta$ and verify that the continuous bifurcation structure is recovered as $\eta\to 0$. This isolates discretization effects from network structure, agent heterogeneity, and learning.

Figure~\ref{fig:discrete_mf_phase} presents the phase diagram in $(\Delta, \eta)$ space. For small $\eta$, the discrete system recovers the ODE bifurcation structure. As $\eta$ increases, the stability margin narrows and the dynamics overshoot the ODE limit-cycle amplitude (Figure~\ref{fig:discrete_mf_trajectories}). Figure~\ref{fig:discrete_mf_eta} quantifies this: the stability margin decreases monotonically with $\eta$. This establishes that discretization introduces additional instability beyond delay. The result provides a baseline for interpreting the full simulation, where discrete time steps are inherent.

\begin{figure}[ht]
\centering
\includegraphics[width=0.75\textwidth]{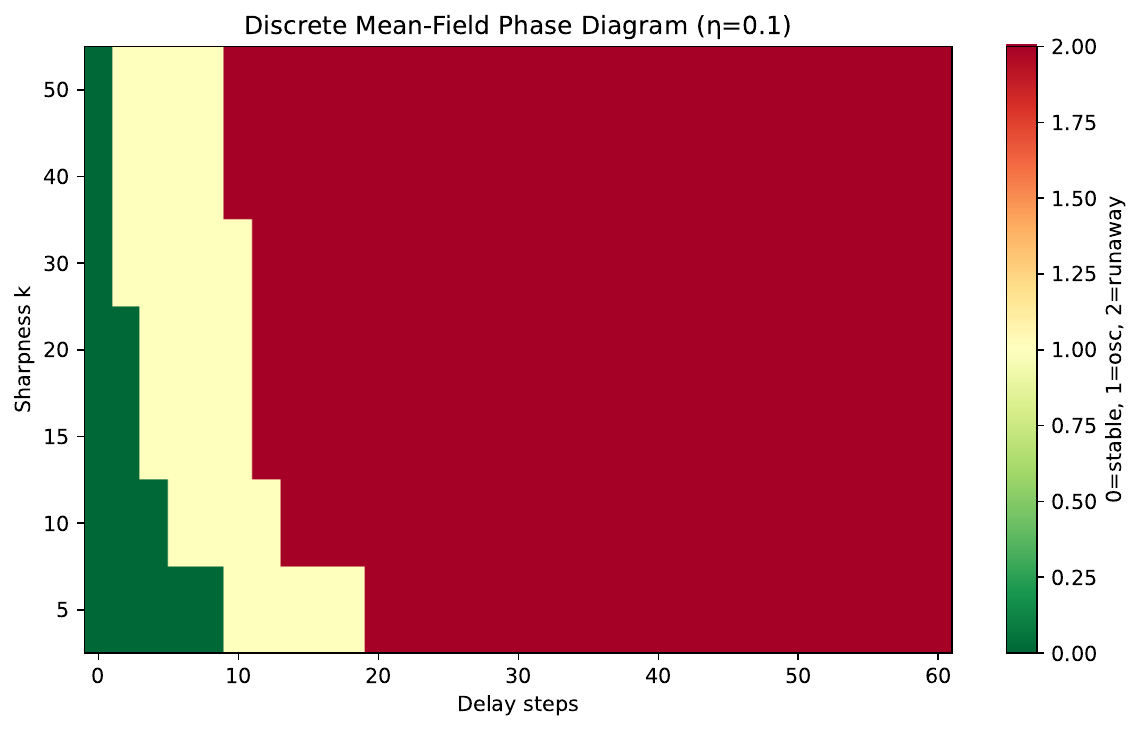}
\caption{Phase diagram of the discrete mean-field system in $(\Delta, \eta)$ space (Experiment~2). Color: regime classification. \textbf{Conclusion:} the stability region contracts with increasing $\eta$, confirming that discretization itself is an additional instability source---the full simulation's discrete time steps make the system \emph{more} fragile than the ODE predicts.}
\label{fig:discrete_mf_phase}
\end{figure}

\begin{figure}[ht]
\centering
\includegraphics[width=0.85\textwidth]{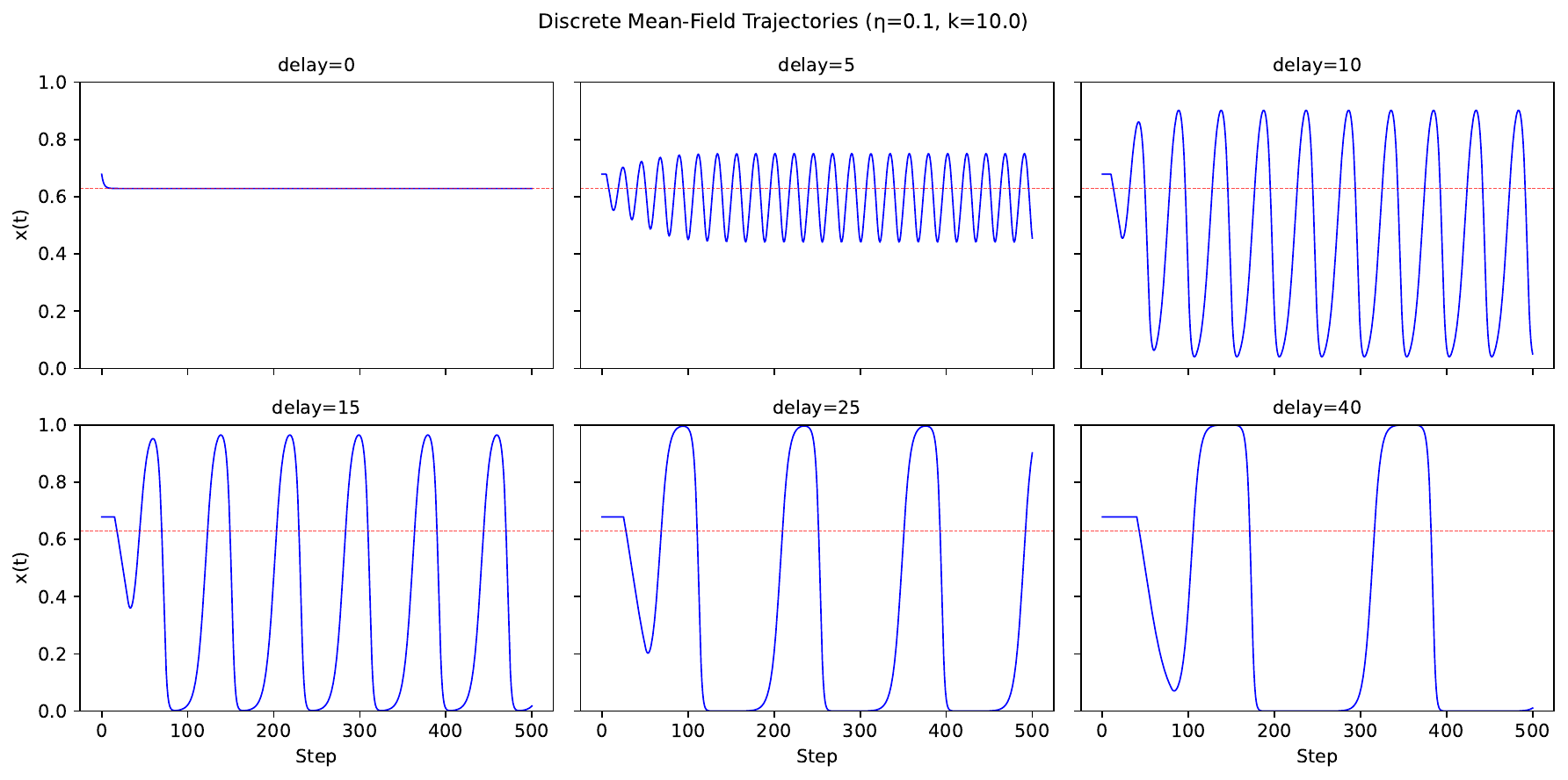}
\caption{Discrete mean-field trajectories at $\Delta{=}1.2\Delta_c$ for several $\eta$ values (Experiment~2). Small $\eta$: bounded oscillations matching the ODE limit cycle. Large $\eta$: overshooting dynamics exceeding the ODE envelope. \textbf{Conclusion:} the full simulation's discrete time steps ($\eta \approx 1$) introduce additional instability beyond the continuous theory.}
\label{fig:discrete_mf_trajectories}
\end{figure}

\begin{figure}[ht]
\centering
\includegraphics[width=0.6\textwidth]{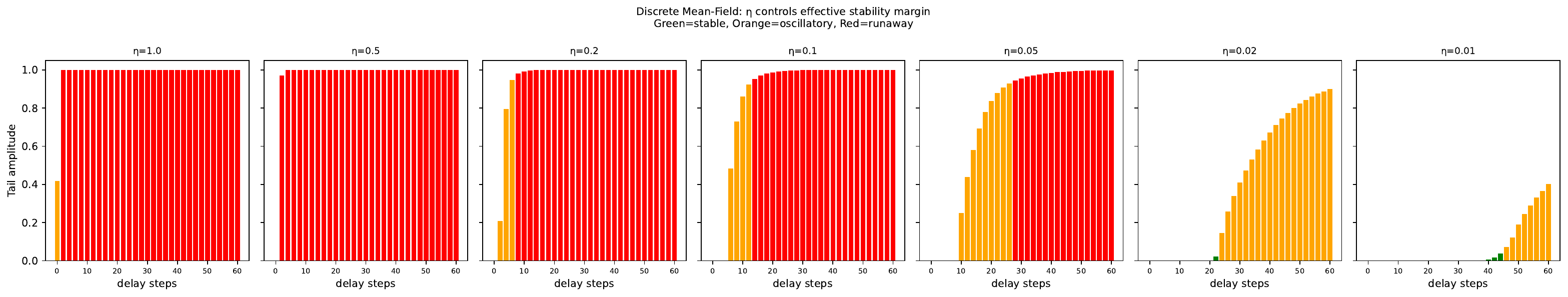}
\caption{Stability margin vs.\ discrete step size $\eta$ (Experiment~2). As $\eta\to 0$, the margin recovers $\Delta_c$ from the continuous theory. \textbf{Conclusion:} discretization reduces the stability margin monotonically. This provides a quantitative bridge between the ODE prediction and the simulation's inherent discrete dynamics.}
\label{fig:discrete_mf_eta}
\end{figure}

\subsection{Dose-response (Experiment~3, Experiment~4)}

\paragraph{Experiment~3: Delay sweep.} This experiment tests Hypothesis~1 (delay destabilizes) in the full networked simulation with Q-learning agents. We sweep repression delay from 0 to 30 steps using $k{=}20$ (placing the system well above $\Delta_c$) and 500 time steps (sufficient for Q-value convergence). If the ODE mechanism survives, runaway frequency should increase monotonically with delay.

Table~\ref{tab:delay_sweep} reports runaway rate vs.\ delay ($k{=}20$, Q-learning agents). At delay${}{=}0$ the system is already only 46\% stable (44\% oscillatory, 10\% runaway): the non-delayed dynamics sit near the oscillatory boundary, and the 10\% runaway reflects stochastic noise rather than delay (Theorem~\ref{thm:nodelay}). Increasing delay shifts this distribution decisively into runaway, which rises to 72\% at delay${}{=}30$ (95\% confidence interval (CI) $[58\%, 83\%]$), an excess of $+62$ percentage points over the delay-zero runaway floor. The monotonic trend confirms the qualitative prediction: delay converts the marginally-stable baseline into large-amplitude runaway. Sixty-two percentage points is not a subtle effect. This sweep uses Q-learning agents, whose adaptive exploration leaves them near the oscillatory boundary even at $\Delta=0$; the cleanest demonstration of destabilization from a \emph{fully} stable baseline is the reactive architecture (Experiment~5 and Figure~\ref{fig:reactive_fine}), which is 100\% stable at $\Delta=0$ and collapses entirely once delay is introduced.

\begin{table}[ht]
\centering
\begin{tabular}{lccc}
\toprule
Delay (steps) & Runaway rate & 95\% CI & Excess over $\Delta{=}0$ (pp) \\
\midrule
0  & 0.10 & $[0.04, 0.21]$ & --- \\
2  & 0.12 & $[0.06, 0.24]$ & $+2$ pp \\
4  & 0.20 & $[0.11, 0.33]$ & $+10$ pp \\
6  & 0.30 & $[0.19, 0.44]$ & $+20$ pp \\
8  & 0.36 & $[0.24, 0.50]$ & $+26$ pp \\
10 & 0.44 & $[0.31, 0.58]$ & $+34$ pp \\
14 & 0.56 & $[0.42, 0.69]$ & $+46$ pp \\
20 & 0.62 & $[0.48, 0.74]$ & $+52$ pp \\
30 & 0.72 & $[0.58, 0.83]$ & $+62$ pp \\
\bottomrule
\end{tabular}
\caption{Runaway rate vs.\ repression delay (Experiment~3; $k{=}20$, $N{=}240$, 500 steps, Q-learning agents, 50 seeds). The 10\% baseline at $\Delta{=}0$ is stochastic noise; the ``Excess'' column isolates the delay contribution. \textbf{Conclusion:} delay increases runaway by 62 percentage points (from 10\% to 72\%), confirming Hypothesis~1---delay alone destabilizes the networked system (question~1). pp = percentage points. Wilson score 95\% CIs reported.}
\label{tab:delay_sweep}
\end{table}

\begin{figure}[ht]
\centering
\includegraphics[width=0.85\textwidth]{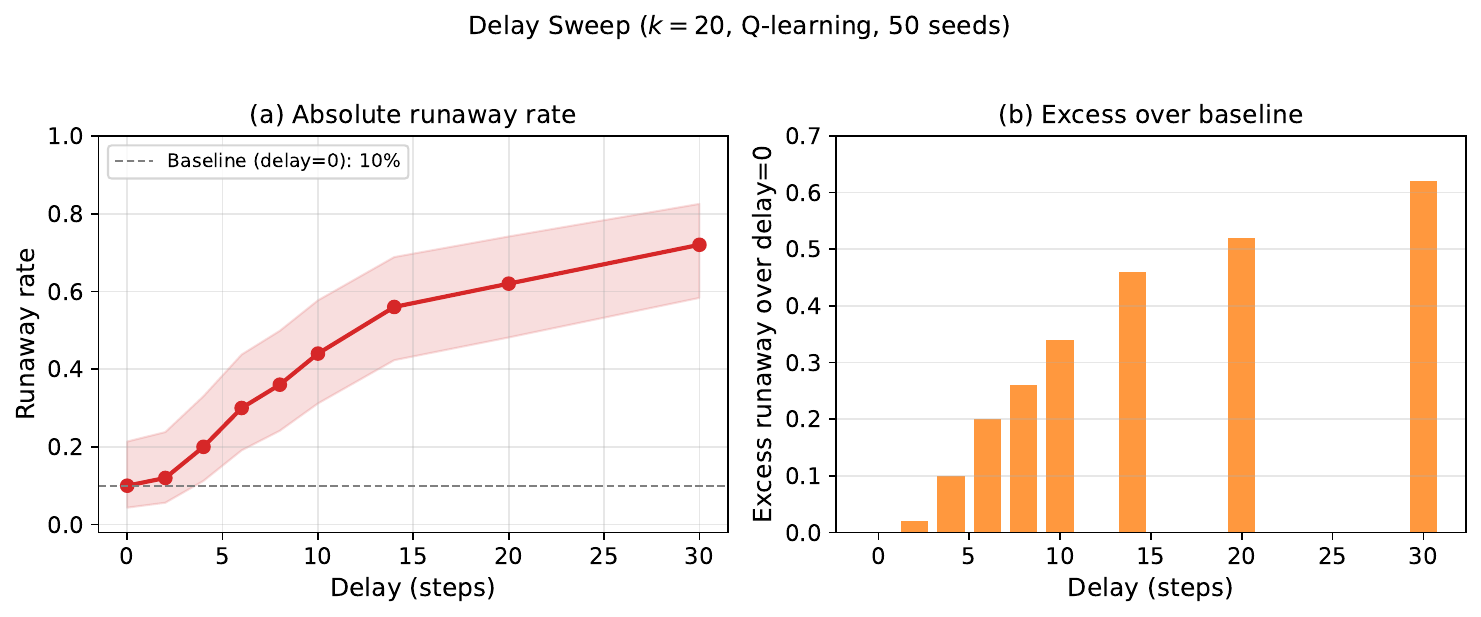}
\caption{Delay sweep (Experiment~3; $k{=}20$, Q-learning, 50 seeds). Left: runaway rate with 95\% Wilson band. Right: excess runaway over the $\Delta{=}0$ baseline. \textbf{Conclusion:} monotonic increase confirms Hypothesis~1: delay alone produces a large, dose-dependent destabilization in the networked simulation.}
\label{fig:delay_sweep}
\end{figure}

\paragraph{Experiment~4: Sharpness sweep.} This experiment tests Hypothesis~2 by holding delay fixed at 15 steps and sweeping $k \in \{3, 5, 7, 10, 15, 20, 30, 40\}$. Higher $k$ should reduce $\Delta_c$ (Equation~\ref{eq:deltac_sigmoid}), so the system should cross from stable to unstable as $k$ increases. Delay of 15 is chosen so that low-$k$ conditions are below threshold while high-$k$ conditions are above it.

At fixed delay${}{=}15$ (Table~\ref{tab:k_sweep}), runaway rate increases from 16\% ($k{=}3$) to 64\% ($k{=}40$). The transition is sharpest between $k{=}7$ (28\%) and $k{=}10$ (52\%), after which the system saturates: once operating far beyond $\Delta_c$, additional sharpening produces only marginal instability gains.

\begin{table}[ht]
\centering
\begin{tabular}{lccc}
\toprule
$k$ & Stable & Oscillatory & Runaway \\
\midrule
3  & 70\% & 14\% & 16\% \\
5  & 56\% & 24\% & 20\% \\
7  & 40\% & 32\% & 28\% \\
10 & 26\% & 22\% & 52\% \\
15 & 26\% & 14\% & 60\% \\
20 & 24\% & 16\% & 60\% \\
30 & 20\% & 18\% & 62\% \\
40 & 24\% & 12\% & 64\% \\
\bottomrule
\end{tabular}
\caption{Regime distribution by sharpness $k$ at fixed delay${}{=}15$ (Experiment~4; $N{=}240$, 500 steps, 50 seeds). \textbf{Conclusion:} runaway increases from 16\% ($k{=}3$) to 64\% ($k{=}40$), with a sharp transition at $k{=}7$--$10$ corresponding to the system crossing $\Delta_c$. This confirms Hypothesis~2 (question~2): sharpness amplifies delay vulnerability.}
\label{tab:k_sweep}
\end{table}

\begin{figure}[ht]
\centering
\includegraphics[width=0.75\textwidth]{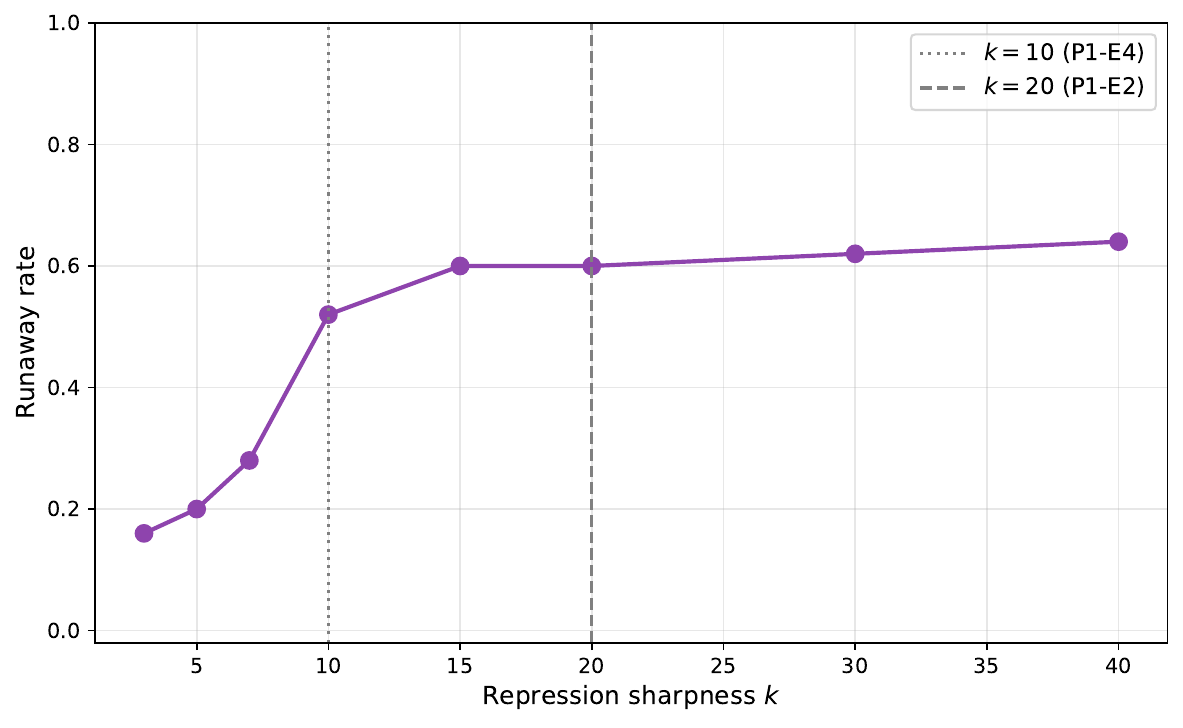}
\caption{Runaway rate vs.\ sharpness $k$ at fixed delay${}{=}15$ (Experiment~4; 50 seeds). Vertical lines: $k{=}10$ (used in Experiment~5) and $k{=}20$ (Experiment~3). \textbf{Conclusion:} the sharp transition at $k{=}7$--$10$ marks the system crossing $\Delta_c$ (Equation~\ref{eq:deltac_sigmoid}). Above this threshold, further sharpening produces diminishing returns, a ceiling effect.}
\label{fig:phase_diagram}
\end{figure}

\subsection{Central experiment: crossed delay $\times$ architecture (Experiment~5)}

This is the central experiment. A two-factor design crosses delay ($\Delta \in \{0, 4, 8, 14, 20\}$) with agent architecture: \emph{fixed-policy} (no environmental response, the null model), \emph{reactive} (threshold heuristic without memory), and \emph{Q-learning} (tabular RL with cumulative value estimates). We use $k{=}10$ to ensure fixed-policy agents remain stable as a clean baseline. The naive expectation is that learning amplifies instability; the alternative is that cumulative learning provides implicit memory that dampens oscillations.

The two-factor crossing reveals a clear hierarchy (Figure~\ref{fig:crossed} and Table~\ref{tab:crossed}).

\begin{figure}[ht]
\centering
\includegraphics[width=0.75\textwidth]{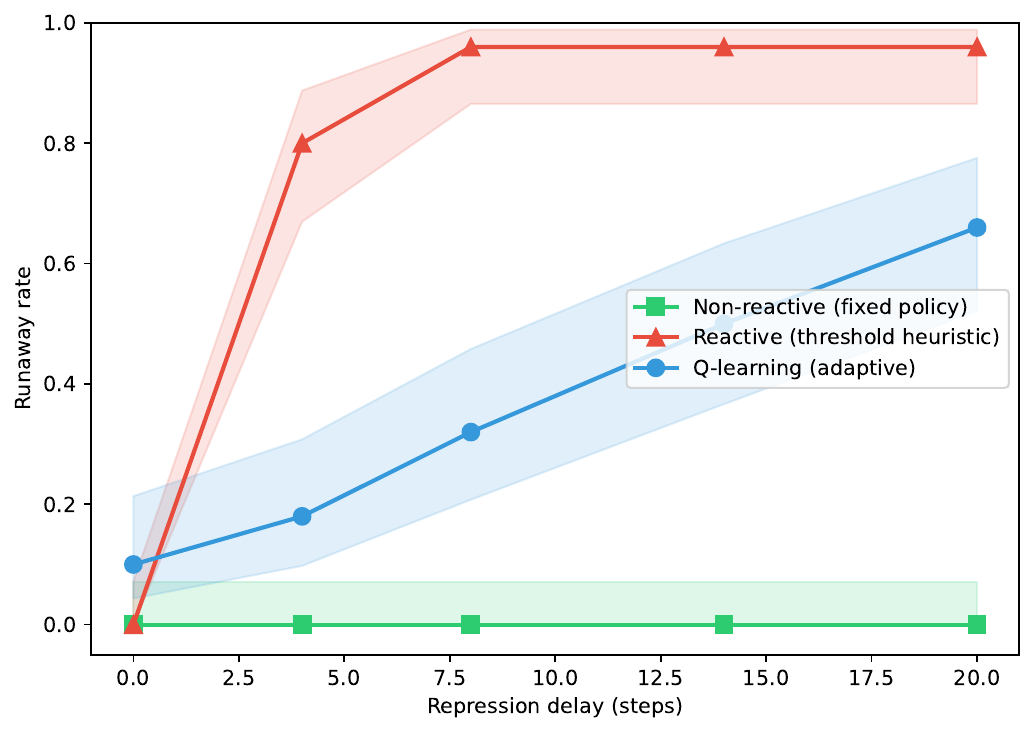}
\caption{Runaway rate vs.\ delay for three agent architectures (Experiment~5, the central experiment; $k{=}10$, $N{=}240$, 500 steps, 50 seeds per cell). Shaded: 95\% Wilson CIs. \textbf{Conclusion:} the hierarchy is reactive (96\%) $>$ Q-learning (66\%) $>$ fixed (0\%). Learning \emph{buffers} rather than amplifies instability (question~3), contradicting the naive expectation. The destabilizing ingredient is memoryless reactivity to delayed signals.}
\label{fig:crossed}
\end{figure}

This experiment uses $k=10$ (lower than the aggressive $k=20$ of Experiment~3) and 500 time steps with $N=240$ agents. The results reveal three distinct regimes:

\begin{enumerate}
\item \textbf{Non-reactive (fixed):} 0\% runaway at all delays (92\% stable, 8\% oscillatory from finite-sample binomial noise in the action draws, occasionally triggering the spectral classifier). These agents cannot exploit temporal structure in the repression signal.
\item \textbf{Reactive (threshold heuristic):} fully stable without delay (100\% stable at delay${}=0$), then a sharp transition to large-amplitude runaway as delay grows: 80\% runaway at delay${}=4$, rising to 96\% at delay${}=8$ (95\% CI $[87\%, 99\%]$). A fine-resolution sweep (Figure~\ref{fig:reactive_fine}) locates the threshold sharply between delay${}=2$ (90\% stable) and delay${}=3$ (62\% runaway), the discrete-time analog of the critical delay $\Delta_c$. This is the cleanest instance of the central claim: an otherwise fully stable population destabilized by delay alone. The reactive policy immediately exploits low-alarm windows but has no memory of past punishment to dampen the resulting oscillations.
\item \textbf{Q-learning:} 10\% runaway at delay${}=0$, rising to 66\% at delay${}=20$ (95\% CI $[52\%, 78\%]$). Q-values encode cumulative punishment experience, which creates inertia against full exploitation of every low-alarm window.
\end{enumerate}

The key finding is the ordering: reactive $>$ Q-learning $>$ fixed in delay sensitivity. The immunity of fixed-policy agents is structural, not contingent on $k$: because their action distribution does not depend on the alarm signal, the feedback loop required for delay-induced instability is broken regardless of parameter values. Reactivity to delayed signals is the destabilizing mechanism; adaptive learning partially buffers this through implicit memory rather than amplifying it.

\begin{table}[ht]
\centering
\begin{tabular}{lcccc}
\toprule
Delay (steps) & Fixed & Reactive & Q-learning & 95\% CI (Q-learning) \\
\midrule
0  & 0\%  & 0\%  & 10\% & $[4\%, 21\%]$ \\
4  & 0\%  & 80\% & 18\% & $[10\%, 31\%]$ \\
8  & 0\%  & 96\% & 32\% & $[21\%, 46\%]$ \\
14 & 0\%  & 96\% & 50\% & $[37\%, 63\%]$ \\
20 & 0\%  & 96\% & 66\% & $[52\%, 78\%]$ \\
\bottomrule
\end{tabular}
\caption{Runaway rate by delay and agent architecture (Experiment~5; $k{=}10$, $N{=}240$, 500 steps, 50 seeds per cell). \textbf{Conclusion:} the ordering fixed (0\%) $<$ Q-learning (up to 66\%) $<$ reactive (up to 96\%) holds at every delay $\Delta \geq 4$; at $\Delta=0$ all three are at the floor (fixed 0\%, reactive 0\%, Q-learning 10\% from exploration noise). Q-learning's implicit punishment memory provides partial protection; memoryless reactivity is catastrophically fragile. This answers question~3: learning buffers, not amplifies, delay-induced instability.}
\label{tab:crossed}
\end{table}

\begin{figure}[ht]
\centering
\includegraphics[width=0.6\textwidth]{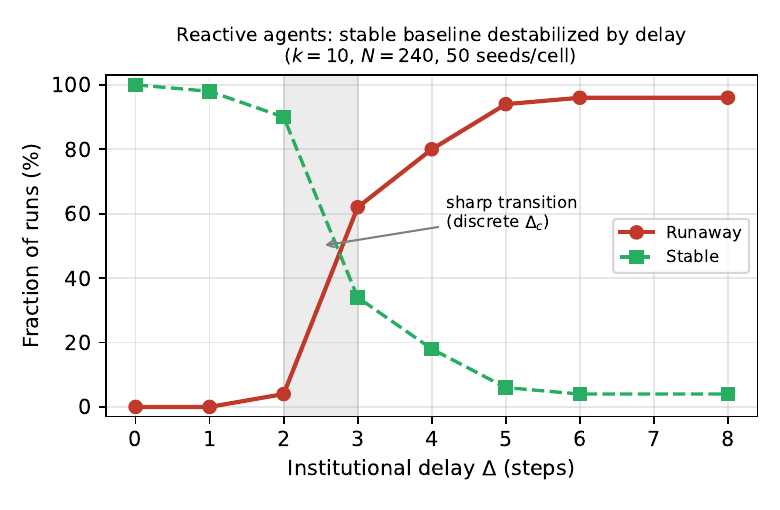}
\caption{Fine-resolution delay sweep for reactive agents ($k{=}10$, $N{=}240$, 50 seeds per delay). The population is fully stable at $\Delta=0$ and through $\Delta=2$ (90\% stable), then transitions sharply to runaway between $\Delta=2$ and $\Delta=3$---the discrete-time analog of the critical delay $\Delta_c$. \textbf{Conclusion:} delay alone, applied to an otherwise stable population, drives the destabilization (question~1).}
\label{fig:reactive_fine}
\end{figure}

\subsection{Exploratory: RL regulator (Experiment~6)}

This experiment compares three governance configurations: fixed-policy agents with a static regulator, Q-learning agents with a static regulator, and Q-learning agents with an RL-trained regulator. The RL regulator selects force $u_t \in \{0, 0.5, 1.0, 1.5, 2.0, 2.5\}$ to minimize $r_{\mathrm{reg}} = -\Rfrac - 0.1\,u_t$, where $\Rfrac$ is the radical fraction and $u_t$ is the control intensity. The regulator shares the institution's information delay, so it cannot circumvent the destabilizing lag. This experiment is exploratory: 50 seeds provide limited statistical power, and we report results as suggestive.

\begin{figure}[ht]
\centering
\includegraphics[width=0.65\textwidth]{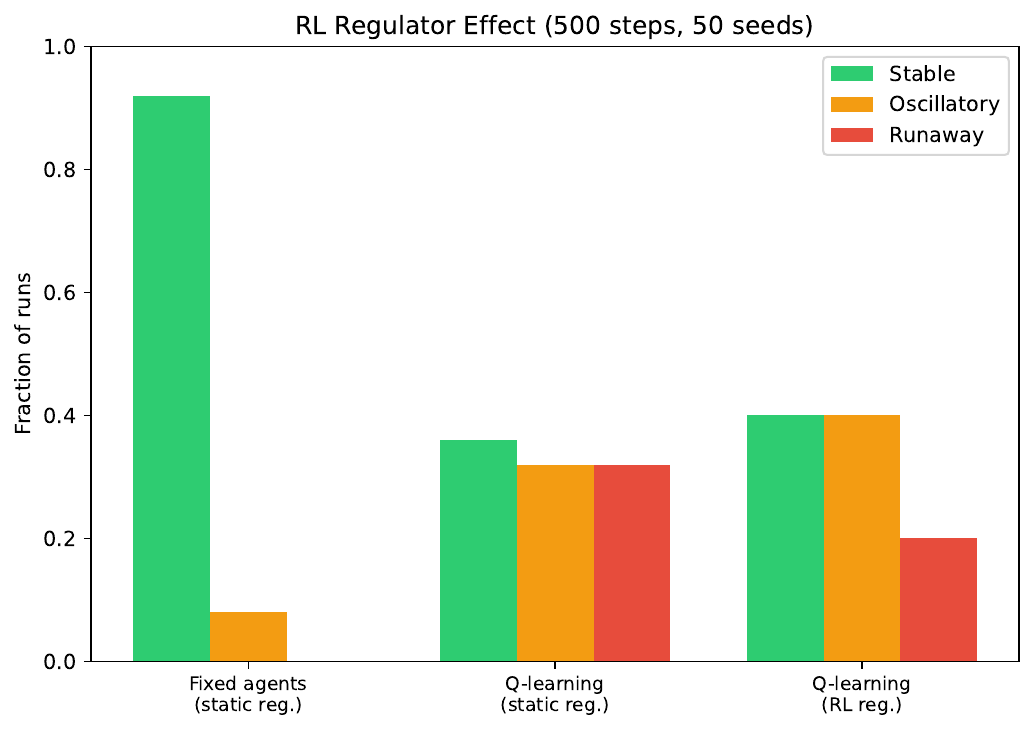}
\caption{Regime distribution for three governance configurations (Experiment~6; 500 steps, 50 seeds). \textbf{Conclusion:} the RL regulator reduces runaway from 32\% to 20\% but increases oscillations from 32\% to 40\%, suggesting mode conversion (disaster into bounded oscillation) rather than full stabilization. The effect is suggestive but not statistically conclusive at 50 seeds (question~4).}
\label{fig:rl_regulator}
\end{figure}

The learning ablation compares three governance configurations (Table~\ref{tab:ablation}). Fixed-policy agents achieve 92\% stability, establishing the baseline. Q-learning agents with a static regulator produce 36\% stable, 32\% oscillatory, and 32\% runaway. Introducing an RL-trained regulator (a Q-learning agent that observes the delayed alarm and radical fraction, selects force $u_t \in \{0, 0.5, 1.0, 1.5, 2.0, 2.5\}$, and receives reward $r_{\mathrm{reg}} = -\Rfrac - 0.1\,u_t$, where $\Rfrac$ is the radical fraction) shifts the distribution to 40\% stable, 40\% oscillatory, and 20\% runaway. This result is suggestive rather than conclusive: the 12 percentage-point reduction in runaway (32\% to 20\%) on 50 seeds yields a confidence interval that includes zero effect. The simultaneous increase in oscillatory outcomes from 32\% to 40\% suggests that if the regulator does help, it converts catastrophic runaway into bounded oscillations rather than restoring stability. We treat this as exploratory evidence that adaptive governance can transform instability modes, not as a confirmed finding.

\begin{table}[ht]
\centering
\begin{tabular}{lccc}
\toprule
Condition & Stable & Oscillatory & Runaway \\
\midrule
Fixed agents (static regulator) & 92\% & 8\% & 0\% \\
Q-learning (static regulator)   & 36\% & 32\% & 32\% \\
Q-learning (RL regulator)       & 40\% & 40\% & 20\% \\
\bottomrule
\end{tabular}
\caption{Regime distribution by governance configuration (Experiment~6; 50 seeds each). \textbf{Conclusion:} adaptive governance converts runaway into oscillation (32\%$\to$20\% runaway, 32\%$\to$40\% oscillatory) rather than restoring stability. This mode conversion is the realistic ceiling for a controller sharing the institution's delay (question~4, exploratory).}
\label{tab:ablation}
\end{table}

\subsection{Summary of findings}

The central result is the two-factor crossing (Experiment~5): reactive agents collapse catastrophically under delay (96\%), Q-learning agents achieve partial resilience (66\%), and non-reactive agents are immune (0\%). Reactivity is the destabilizing mechanism; learning buffers it. The story is simpler than one might expect. The ODE validation (Experiment~1) and discrete mean-field (Experiment~2) confirm that the analytical mechanism is correct and survives discretization. The delay sweep (Experiment~3) and sharpness sweep (Experiment~4) map the dose-response and activation boundary. The RL regulator (Experiment~6) is exploratory: suggestive of mode conversion but not statistically conclusive.

All regime classifications use the adaptive classifier with runaway threshold $R_{\max}>0.40$. The quantitative gap between mean-field thresholds and simulation thresholds reflects the expected difference between an analytically tractable mechanism and a system with agent heterogeneity, stochastic exploration, and network structure.

%% file: sections/06_discussion.tex
\section{Discussion}

The results confirm the core analytical prediction (delay destabilizes) and reveal a surprising empirical finding (learning buffers rather than amplifies). This section discusses three qualitative insights that go beyond the numerical results: why the ODE-to-simulation gap exists and what it teaches us about the mechanism, what the architecture hierarchy reveals about the nature of delay vulnerability, and what practical implications follow for institutional design.

\subsection{Why the quantitative gap between theory and simulation matters (Q1)}

The ODE predicts a sharp bifurcation at $\Delta_c$; the simulation shows a gradual dose-response curve. This gap is not a failure of the theory---it is informative about the mechanism. Three stabilizing buffers are absent from the mean-field model and present in the simulation: stochastic exploration ($\epsilon$-greedy noise damps incipient oscillations), the three-action space (agents can retreat to moderate behavior rather than oscillating between extremes), and spatial heterogeneity (network structure disperses the synchronized oscillations required for large-amplitude instability). These are not defects of the simulation; they are what the mean-field theory trades away for a closed-form answer.

The qualitative lesson is that delay-induced instability is robust to realistic complications, but the sharp threshold of the ODE becomes a gradual erosion of stability in a heterogeneous population. Institutions operating near the theoretical boundary should not expect a clean phase transition; they should expect slowly worsening oscillatory behavior that may be difficult to distinguish from normal fluctuations until runaway is already underway.

\subsection{What the architecture hierarchy reveals about the mechanism (Q3)}

The ordering (reactive agents most fragile, Q-learning agents partially resilient, fixed-policy agents immune) points to a specific causal mechanism rather than a generic correlation between complexity and instability.

The mechanism is a feedback loop: when the alarm drops (due to institutional processing lag), reactive agents immediately escalate. Punishment arrives too late; by then the alarm spike triggers the next oscillation. The critical feature is not that reactive agents are ``simple'' but that they are \emph{memoryless}: each decision is based entirely on the current (stale) signal, with no integration over past experience. Q-learning agents break this loop because their Q-values carry forward the memory of past punishment. The ``soft brake'' is not a design feature---it is an emergent consequence of temporal-difference learning applied to a delayed-feedback environment.

This distinction has a practical implication: the risk factor for delay vulnerability is not whether agents are adaptive, but whether they integrate information over time. Any decision process that reacts only to current signals (whether a simple threshold, a rule-based policy, or a sophisticated model with no memory state) will be vulnerable. Institutions seeking stability should therefore prioritize reducing processing delay over restricting agent adaptiveness. The problem is the lag, not the learning.

\subsection{Policy levers: response sharpness and adaptive governance (Q2, Q4)}

The theory predicts that gradual institutional responses (low $k$) expand the stability margin. The simulation confirms this, but with an important qualification: the protective effect of gradual response is conditional on agent architecture. Fixed-policy agents are stable regardless of sharpness, because they do not close the feedback loop. For adaptive agents, reducing sharpness helps but cannot eliminate delay vulnerability---it only raises the threshold. The policy implication is that institutions face a genuine dilemma: sharp responses are decisive but fragile; gradual responses are resilient but slow to act. Every parameter trades one vulnerability for another.

The second lever is the regulator itself. The RL regulator experiment is exploratory, and the specific numbers should be interpreted cautiously. The qualitative pattern, however, is suggestive: the regulator appears to convert catastrophic runaway into bounded oscillations rather than restoring full stability. This mode conversion---disaster into nuisance---may be the realistic ceiling for any controller that shares the institution's information delay. A regulator that observes the same stale signal as the repression mechanism cannot anticipate the population's response; it can only react to the consequences of its own delayed actions, introducing a secondary feedback loop that sustains limit cycles while preventing unbounded growth. We report this as a pattern worth investigating, not a confirmed finding: 50 seeds is enough to see a pattern but not enough to bet on it.

\subsection{Limitations}

This model is deliberately simple, and the simplifications cost something. The Q-learning model is deliberately minimal and does not capture sophisticated strategic reasoning, communication, or coordination among agents. Real adaptive agents in social or technical systems may exhibit richer behavioral repertoires. The three-action space is a simplification; continuous action spaces might produce different stability characteristics. The regime classifier uses fixed thresholds, and the precise quantitative rates depend on these choices. We verified that the central ordering (fixed $\leq$ Q-learning $\leq$ reactive in delay sensitivity) is preserved across runaway thresholds from $R_{\max}>0.30$ through $R_{\max}>0.50$, though the absolute rates shift substantially (e.g., reactive at delay${}=20$ ranges from 100\% at threshold 0.30 to 64\% at threshold 0.50). A convergence diagnostic confirms that Q-learning agents reach stable behavior by step 100--200: the mean radical fraction and its standard deviation change by less than 0.001 between the last two 100-step windows at all tested delays, indicating that 500 steps is sufficient for the Q-values to equilibrate.

The simulation uses finite populations ($N=240$) on specific graph realizations; larger populations or different graph families might shift the quantitative boundaries. The architecture hierarchy itself, however, is robust to the network parameterization: re-running the crossed delay~$\times$~architecture design on a denser modular graph (4 communities, $p_{\mathrm{in}}=0.15$, $p_{\mathrm{out}}=0.02$) reproduces the same ordering---fixed agents immune (0\% at all delays), reactive agents catastrophic (100\% by delay~8), Q-learning agents partially resilient (77\% at delay~20)---with the denser network somewhat more unstable (reactive reaches 27\% runaway even at delay~0), consistent with its stronger influence coupling.

The reactive baseline uses a single threshold heuristic; other reactive architectures (imitation dynamics, best-response with noise) might show different levels of delay sensitivity, though the qualitative finding (that memoryless reactivity is more fragile than learning) should hold for any policy that lacks temporal integration. The RL regulator's Q-update bootstraps toward a next-state computed from the current (undelayed) alarm rather than the next delayed observation, introducing a minor state-transition inconsistency; since Experiment~6 is presented as exploratory evidence, this does not affect the paper's confirmed findings. Finally, the theory assumes a single aggregate alarm signal; systems with multiple observation channels or local feedback might exhibit different instability structures. Reality, of course, does not present its instabilities one at a time.

\subsection{Relationship to prior work}

Our analytical results connect directly to the Hopf bifurcation analysis of \citet{wesson2016hopf} for two-strategy delayed replicator dynamics, extending their framework to include an asymmetric institutional response function rather than symmetric frequency-dependent payoffs. The discrete-time instability amplification we observe in Experiment~2 echoes the finding of \citet{alboszta2004stability} that delay can destabilize the ESS in discrete replicator dynamics; \citet{iijima2012delayed} documents the complementary regime in which stability is instead preserved, underscoring that the destabilizing effect of delay is model-dependent. The network simulation builds on the evolutionary-dynamics-on-graphs tradition \citep{lieberman2005evolutionary, ohtsuki2006simple, szabo2007evolutionary}, adding delayed institutional feedback as a novel destabilizing mechanism distinct from the structural effects studied in that literature. The interaction between learning and delay connects to recent work on reward delays in multi-agent reinforcement learning \citep{zhang2023marl}, though our focus on population-level regime transitions rather than individual convergence is novel.

The companion paper extends the present mechanism to noisy selective control on modular networks, addressing the question of how imperfect observation and targeted governance alter stability conditions. That extension is necessary because real institutions rarely observe true system state directly, and the consequences of classification errors depend strongly on network position, a consideration we deliberately set aside here so that delay gets a fair hearing on its own.

%% file: sections/07_conclusion.tex
\section{Conclusion}

We studied how institutional processing delay affects the stability of a multi-agent system in which agents adapt their behavior in response to delayed punishment signals. We derived a closed-form critical delay $\Delta_c$ for the delayed replicator equation with a sigmoid response function and proved that the resulting Hopf bifurcation is supercritical for the entire admissible parameter class. We then tested three agent architectures (fixed-policy, reactive, and Q-learning) in a networked simulation with 240 agents across 50 seeds per condition.

The central result answers question~3: learning does not amplify delay-induced instability---it partially buffers it. Non-reactive agents are immune to delay (0\% runaway), reactive agents collapse catastrophically (96\%), and Q-learning agents achieve partial resilience (66\%). The delay sweep (question~1) confirms a monotonic dose-response reaching $+62$ percentage points of excess runaway, the sharpness sweep (question~2) confirms that sharper institutional responses lower the stability threshold, and the RL regulator experiment (question~4) suggests that adaptive governance converts runaway into bounded oscillations rather than restoring stability. The destabilizing ingredient, then, is memoryless reactivity to delayed signals, not learning: agents that immediately exploit low-alarm windows trigger oscillatory feedback loops, while agents with cumulative memory resist this trap. The practical implication is twofold. Institutions should prefer gradual responses over sharp thresholds, because sharpness amplifies the instability that reactivity exploits; and adaptive memory is not the disease but the closest thing to a cure this system has.

Three extensions are natural: scaling to larger populations and alternative network families to test how far the architecture hierarchy holds; continuous action spaces, which may produce different stability characteristics; and combining delay with noisy observation, which the companion paper addresses for selective governance on modular networks.

%% file: sections/appendix_hopf_proof.tex
\section{Proof of Supercritical Hopf Bifurcation}
\label{app:hopf_proof}

This appendix establishes Proposition~\ref{prop:supercritical}: the Hopf bifurcation at $\Delta=\Delta_c$ is supercritical for all admissible sigmoid parameters. The proof follows the Hassard--Kazarinoff--Wan center manifold reduction \citep{hassard1981theory} applied to the DDE phase space \citep{hale1993introduction}. We proceed in five steps, summarized in Table~\ref{tab:proof_roadmap}.

\begin{table}[ht]
\centering
\begin{tabular}{clll}
\toprule
Step & Goal & Key result & Equation \\
\midrule
1 & Linearize and set up phase space & Hayes equation, $\beta=-b$ & \eqref{eq:delayed_rep} \\
2 & Eigenvectors and adjoint normalization & $\bar{D}=1/(1+i\pi/2)$ & \\
3 & Nonlinear expansion to cubic order & $f_2$, $f_3$ & \eqref{eq:f2f3} \\
4 & Center manifold and normal form & $c_1(0)$ & \eqref{eq:c1_appendix} \\
5 & Sign analysis of $\operatorname{Re}(c_1)$ & $\operatorname{Re}(c_1)<0$ universally & \eqref{eq:Re_c1_closed_app} \\
\bottomrule
\end{tabular}
\caption{Roadmap of the center-manifold reduction for Proposition~\ref{prop:supercritical}. \textbf{Conclusion:} the five steps culminate in $\operatorname{Re}(c_1)<0$ for all admissible sigmoid parameters, establishing that the Hopf bifurcation is supercritical (stable bounded limit cycles rather than explosive growth).}
\label{tab:proof_roadmap}
\end{table}

\paragraph{Notation.} Throughout this appendix: $\rho = p(x^*) = a/C \in (0,1)$ is the equilibrium repression probability, $\alpha_0 = x^*(1-x^*)$, $\alpha_1 = 1-2x^*$, and $b = C\alpha_0 p'(x^*)$. All derivatives of $p$ are evaluated at $x^*$.

\begin{proof}[Proof of Proposition~\ref{prop:supercritical}]

\medskip
\noindent\textbf{Step 1: Linearization and phase space.}
Write $u(t)=x(t)-x^*$ and define $F(u,u_d)$ as the full nonlinearity, where $u_d=u(t-\Delta)$. Since $a-Cp(x^*)=0$ at equilibrium, the instantaneous linearization vanishes: $\partial_u F(0,0)=0$. The linearized equation is
\[
\dot{u}(t)=\beta\,u(t-\Delta), \qquad \beta=-b<0,
\]
the Hayes equation, whose phase space is $C=C([-\Delta_c,0],\mathbb{R})$.

\medskip
\noindent\textbf{Step 2: Eigenvectors and adjoint normalization.}
At $\Delta=\Delta_c$, the characteristic equation $\lambda+be^{-\lambda\Delta}=0$ has roots $\lambda=\pm i\omega$ with $\omega=b$. Define:
\begin{align*}
\text{Eigenvector:} \quad & \varphi(\theta)=e^{i\omega\theta}, \quad \theta\in[-\Delta_c,0]. \\
\text{Adjoint eigenvector:} \quad & \psi(s)=\bar{D}\,e^{i\omega s}, \quad s\in[0,\Delta_c].
\end{align*}
Normalizing via the Hale--Verduyn Lunel bilinear form $\langle\psi,\varphi\rangle=1$ yields
\[
\bar{D}=\frac{1}{1+i\pi/2}.
\]

\medskip
\noindent\textbf{Step 3: Nonlinear expansion to cubic order.}
The full nonlinearity is
\[
F(u,u_d) = (\alpha_0 + \alpha_1 u - u^2)\bigl[-Cp'u_d - \tfrac{1}{2}Cp''u_d^2 - \tfrac{1}{6}Cp'''u_d^3 + \cdots\bigr],
\]
where $\alpha_0 = x^*(1{-}x^*)$ and $\alpha_1 = 1{-}2x^*$. Collecting by order:
\begin{equation}
\label{eq:f2f3}
\begin{aligned}
f_2(u,u_d) &= -C\alpha_1 p'\; u\,u_d
              &&-\; \tfrac{1}{2}\,C\alpha_0 p''\; u_d^2, \\[4pt]
f_3(u,u_d) &= \phantom{-}C p'\; u^2 u_d
              &&-\; \tfrac{1}{2}\,C\alpha_1 p''\; u\,u_d^2
              \;-\; \tfrac{1}{6}\,C\alpha_0 p'''\; u_d^3.
\end{aligned}
\end{equation}
For the sigmoid response, the derivatives at $x^*$ evaluate to:
\begin{align}
p'  &= k\,\rho(1-\rho), \nonumber\\
p'' &= k^2\,\rho(1-\rho)(1-2\rho), \label{eq:sigmoid_derivs}\\
p'''&= k^3\,\rho(1-\rho)\bigl(1-6\rho(1-\rho)\bigr). \nonumber
\end{align}

\medskip
\noindent\textbf{Step 4: Center manifold reduction and normal form.}
On the center manifold, $u\approx z+\bar{z}$ and $u_d\approx i(\bar{z}-z)$ (since $e^{-i\omega\Delta_c}=-i$). Projecting $f_2$ onto the center eigenspace gives the second-order normal form coefficients $g_{20}$, $g_{11}$, $g_{02}$.

The center manifold corrections $W_{20}(\theta)$ and $W_{11}(\theta)$ satisfy the boundary value problems
\begin{align*}
(2i\omega - \mathcal{A})\,W_{20} &= H_{20}, \\
-\mathcal{A}\,W_{11} &= H_{11},
\end{align*}
where $\mathcal{A}$ is the infinitesimal generator of the linearized semigroup. The right-hand sides $H_{20}$, $H_{11}$ are determined by projecting the quadratic nonlinearity onto the stable complement. The resulting values at $\theta=0$ and $\theta=-\Delta_c$ are verified symbolically (supplementary code) to satisfy their boundary conditions identically for all parameter values.

These corrections yield the third-order coefficient $g_{21}$ and the first Lyapunov coefficient:
\begin{equation}
c_1(0) \;=\; \frac{i}{2\omega}
  \Bigl(\,g_{20}\,g_{11} - 2\,|g_{11}|^2 - \tfrac{1}{3}\,|g_{02}|^2\,\Bigr)
  \;+\; \frac{g_{21}}{2}.
\label{eq:c1_appendix}
\end{equation}

\medskip
\noindent\textbf{Step 5: Sign analysis of $\operatorname{Re}(c_1)$.}
Carrying out the full computation and rationalizing all complex denominators (supplementary code provides the symbolic derivation), we obtain:
\begin{equation}
\operatorname{Re}(c_1(0))=\underbrace{\frac{Ck\rho}{5\alpha_0(4+\pi^2)}}_{>0}\cdot\;\underbrace{\mathcal{N}(k,\rho,\alpha_0,\alpha_1)}_{\text{sign?}},
\label{eq:Re_c1_closed_app}
\end{equation}
where the numerator factors as $\mathcal{N}=(\rho-1)\cdot\mathcal{B}$, with $\rho-1<0$ for all $\rho\in(0,1)$. It remains to show $\mathcal{B}>0$.
\end{proof}

\begin{lemma}[Positivity of $\mathcal{B}$]
\label{lem:B_positive}
For all $\rho\in(0,1)$, $k>0$, and $x^*\in(0,1)$, define
\begin{align*}
\mathcal{B} \;=\;&
  \underbrace{2\,Q(\rho)\;(\alpha_0 k)^2}_{\text{term } A}
  \;-\; \underbrace{(7\pi{-}8)(2\rho{-}1)\;(\alpha_0 k)\,\alpha_1}_{\text{cross-term } B}  \\[4pt]
  &+\; \underbrace{2(3\pi{-}2)\;\alpha_1^2}_{\text{term } C}
  \;+\; \underbrace{10\pi\,\alpha_0}_{\text{remainder}},
\end{align*}
where $Q(\rho)=-(7\pi-8)\rho(1-\rho)+(3\pi-2)$. Then $\mathcal{B}>0$.
\end{lemma}

\begin{proof}
Since $\rho(1-\rho)\le 1/4$, we have $Q(\rho)\ge (3\pi-2)-(7\pi-8)/4 = 5\pi/4 > 0$.

The first three terms form a quadratic form in the variables $v_1=\alpha_0 k$ and $v_2=\alpha_1$:
\[
\mathcal{B}_0(v_1,v_2) = 2Q\,v_1^2 - (7\pi{-}8)(2\rho{-}1)\,v_1 v_2 + 2(3\pi{-}2)\,v_2^2.
\]
This is positive definite whenever the discriminant condition $4AC > B^2$ holds, i.e.,
\[
16\,Q\,(3\pi-2) > (7\pi-8)^2(2\rho-1)^2.
\]
Substituting $\mu=\rho(1-\rho)$ so that $(2\rho-1)^2=1-4\mu$, the condition becomes
\[
[-20\pi(7\pi-8)]\,\mu + 95\pi^2-80\pi > 0.
\]
Since the coefficient of $\mu$ is negative, the worst case is $\mu=1/4$ (i.e., $\rho=1/2$), giving $20\pi(3\pi-2)>0$. Hence $\mathcal{B}_0$ is positive definite for all $\rho\in(0,1)$.

Adding the strictly positive remainder $10\pi\alpha_0>0$ gives $\mathcal{B} = \mathcal{B}_0 + 10\pi\alpha_0 > 0$.
\end{proof}

\noindent\textbf{Conclusion of the proof.}
Since $\rho-1<0$ and $\mathcal{B}>0$ (Lemma~\ref{lem:B_positive}), we have $\mathcal{N}=(\rho-1)\mathcal{B}<0$. The prefactor in \eqref{eq:Re_c1_closed_app} is strictly positive. Therefore $\operatorname{Re}(c_1(0))<0$ for all admissible parameters.

\medskip

\paragraph{Bifurcation direction and orbital stability.} The standard Hassard--Kazarinoff--Wan quantities are
\[
\mu_2 = -\frac{\operatorname{Re}(c_1(0))}{\operatorname{Re}(d\lambda/d\Delta)\big|_{\Delta_c}} > 0, \qquad
\beta_2 = 2\operatorname{Re}(c_1(0)) < 0.
\]
Since $\mu_2>0$, the bifurcation is supercritical (periodic orbits exist for $\Delta>\Delta_c$). Since $\beta_2<0$, the bifurcating periodic orbits are orbitally stable.

\paragraph{Numerical validation.}
\begin{table}[ht]
\centering
\begin{tabular}{lrl}
\toprule
Quantity & Value & Interpretation \\
\midrule
$\operatorname{Re}(c_1(0))$ & $-40.85$ & $<0$: supercritical \\
$\mu_2$ & $15.95$ & $>0$: orbits exist above $\Delta_c$ \\
$\beta_2$ & $-81.70$ & $<0$: orbits are stable \\
\bottomrule
\end{tabular}
\caption{Bifurcation quantities at representative parameters $(a,C,k,x_c)=(2,5,10,0.5)$. \textbf{Conclusion:} $\operatorname{Re}(c_1)<0$, $\mu_2>0$, and $\beta_2<0$ jointly confirm a supercritical Hopf bifurcation, with stable limit cycles emerging for $\Delta>\Delta_c$.}
\label{tab:hopf_numerics}
\end{table}

Three independent checks validate the algebraic reduction:
\begin{enumerate}
\item \textbf{Grid evaluation:} on a $30\times60\times30$ grid (30 values of $\rho$ uniform in $(0,1)$, 60 values of $k$ log-uniform in $[0.1,10^4]$, 30 values of $x_c$ uniform in $[0.02,0.98]$), $\operatorname{Re}(c_1)$ is strictly negative at every \emph{admissible} point---those with an interior equilibrium $x^*\in(0,1)$, which is the hypothesis of the proposition. Grid points with no interior equilibrium fall outside the proposition's scope and are excluded.
\item \textbf{DDE integration:} Direct numerical integration confirms continuous amplitude growth from zero above $\Delta_c$ for all tested parameter sets.
\item \textbf{BVP residuals:} Boundary condition residuals for $W_{20}$ and $W_{11}$ remain below $10^{-12}$ across $10^3$ random parameter draws.
\end{enumerate}
Supplementary scripts (Python/SymPy and Wolfram Language) reproduce all symbolic derivations and numerical checks.

%% file: sections/appendix_experiment_details.tex
\section{Experiment Details}
\label{app:experiments}

This appendix provides implementation details for all experiments reported in the main text. Full experiment configurations are stored as JSON manifests in \texttt{experiments/configs/} and each run produces a deterministic output bundle.

\subsection{Output structure}

Each experimental run produces three artifacts: a \texttt{history.csv} file containing the full time series of behavioral fractions, alarm, repression probability, and punishment counts at each time step; a \texttt{summary.json} file containing aggregate statistics computed over the final 50\% of the simulation horizon and a regime label; and a \texttt{config\_resolved.json} file recording the exact parameter values used. Each sweep directory additionally produces a \texttt{summary.csv} aggregating all per-run summaries. The regime label is assigned by the adaptive classifier (spectral analysis with runaway threshold $R_{\max}>0.40$); the old fixed-threshold classifier is also recorded as \texttt{regime\_fixed} for comparison.

\subsection{Common parameters}

Unless otherwise noted, all networked simulations use $N=240$ agents on a stochastic block model graph with 6 communities, intra-community edge probability $p_{\mathrm{in}}=0.08$, inter-community edge probability $p_{\mathrm{out}}=0.004$, and bridge fraction 12\% (defined by betweenness centrality). Agent Q-learning parameters: $\alpha=0.10$, $\epsilon=0.08$, $\gamma=0.95$. Influence coupling strength $\lambda=0.22$; alarm threshold $A_c=0.72$. Content types are disabled for all Paper~1 experiments (\texttt{enable\_content=false}; agents receive neutral payoffs with no content-type bonuses). All experiments use 50 seeds (1--50).

\subsection{Experiment catalog}

Table~\ref{tab:experiment_summary} summarizes all six experiments at a glance; the paragraphs that follow give the full configuration for each.

\begin{table}[ht]
\centering
\footnotesize
\setlength{\tabcolsep}{4pt}
\begin{tabular}{@{}clp{2.2cm}p{2.3cm}cp{3.2cm}@{}}
\toprule
Exp. & System & Varied factor & Fixed parameters & Runs & Headline finding \\
\midrule
1 & Delayed replicator ODE & $\Delta/\Delta_c \in [0,4]$ & $a{=}2$, $C{=}7$, $k{=}10$, $x_c{=}0.72$ & 80 & Hopf bifurcation at $\Delta_c$ exactly; supercritical \\
2 & Discrete mean-field & step size $\eta$, $\Delta$ & same as Exp.~1 & grid & discretization shrinks the stability margin monotonically \\
3 & Network, Q-learning & $\Delta \in \{0\text{--}30\}$ & $k{=}20$, $N{=}240$ & 450 & runaway $10\% \to 72\%$ ($+62$ pp) \\
4 & Network, Q-learning & $k \in \{3\text{--}40\}$ & $\Delta{=}15$, $N{=}240$ & 400 & runaway $16\% \to 64\%$; knee at $k{=}7\text{--}10$ \\
5 & Network, crossed & $\Delta \times$ architecture & $k{=}10$, $N{=}240$ & 750 & reactive $96\% >$ Q-learning $66\% >$ fixed $0\%$ \\
6 & Network, governance & regulator type & $\Delta{=}6$, $k{=}10$ & 150 & RL regulator converts runaway to oscillation (exploratory) \\
\bottomrule
\end{tabular}
\caption{Summary of all Paper~1 experiments. Experiments~1--2 validate the analytical theory (the ODE and its discretization); Experiments~3--6 test it in the networked simulation with $N{=}240$ agents and 50 seeds per condition. ``Runs'' counts independent simulations: Experiment~1 counts ODE integrations and Experiment~2 is a deterministic $\Delta\times\eta$ grid, while Experiments~3--6 are seeded ($9{\times}50$, $8{\times}50$, $3{\times}5{\times}50$, $3{\times}50$). \textbf{Conclusion:} the delay-instability mechanism survives every added layer of realism, and the central crossing (Experiment~5) isolates reactivity---not learning---as the destabilizing ingredient. All headline numbers reproduce from the frozen \texttt{summary.csv} files under \texttt{experiments/results/paper1\_v2/} (for Experiment~5, group by the \texttt{learning\_mode} column).}
\label{tab:experiment_summary}
\end{table}

\paragraph{Experiment~1: ODE validation.} Numerical integration of the delayed replicator equation~\eqref{eq:delayed_rep} using Euler stepping for the DDE (global truncation error $O(\mathrm{d}t)$; at $\mathrm{d}t=0.01$ the bifurcation point is resolved to within 1\% of $\Delta_c$, verified by halving the step size). Parameters: $a=2.0$, $C=7.0$, $k=10$, $x_c=0.72$. Delay swept from $0$ to $4\Delta_c$ in 80 values. Integration horizon: 120 time units with step size $\mathrm{d}t=0.01$. Output: bifurcation diagram, representative trajectories, and sharpness sweep.

\paragraph{Experiment~2: Discrete mean-field.} Euler discretization of the delayed replicator equation with step size $\eta$ swept over $\{0.01, 0.02, 0.05, 0.1, 0.2, 0.5, 1.0\}$. Same base parameters as Experiment~1. Delay swept from 0 to 60 steps jointly with $\eta$ to produce a phase diagram. Each condition run for 2000 steps with deterministic initial conditions at $x^*+0.05$.

\paragraph{Experiment~3: Delay sweep.} Full networked simulation with common parameters above plus $k=20$, 500 time steps. Delay swept over $\{0, 2, 4, 6, 8, 10, 14, 20, 30\}$ steps. Static regulator with force $u_t=1.0$. Seeds: 50 per delay level.

\paragraph{Experiment~4: Sharpness sweep.} Same base configuration as Experiment~3 with delay fixed at 15 steps. Sharpness $k$ swept over $\{3, 5, 7, 10, 15, 20, 30, 40\}$. Seeds: 50 per sharpness level.

\paragraph{Experiment~5: Crossed delay $\times$ architecture.} Two-factor design crossing delay $\in\{0, 4, 8, 14, 20\}$ with agent architecture $\in\{\text{fixed}, \text{reactive}, \text{Q-learning}\}$. Fixed agents use stationary action weights $[0.65, 0.27, 0.08]$ for $[L, M, R]$, chosen to approximate the empirical radical fraction at Q-learning convergence without delay and thus provide a matched non-adaptive baseline. Reactive agents use a threshold heuristic that responds to the delayed alarm, local influence signal, and punishment state but maintains no memory across steps. Q-learning agents use $\alpha=0.10$, $\epsilon=0.08$. Default sharpness $k=10$. Simulation horizon: 500 steps. Seeds: 50 per cell.

\paragraph{Experiment~6: RL regulator.} Three conditions: (1) fixed agents with static regulator ($u_t=1.0$), (2) Q-learning agents with static regulator, (3) Q-learning agents with RL regulator. The RL regulator is a tabular Q-learning agent ($\alpha_{\mathrm{reg}}=0.05$, $\epsilon_{\mathrm{reg}}=0.10$, $\gamma_{\mathrm{reg}}=0.95$) observing a discretized state (4 alarm buckets $\times$ 3 radical-fraction buckets) and selecting $u_t$ from $\{0.0, 0.5, 1.0, 1.5, 2.0, 2.5\}$. Regulator reward: $r_{\mathrm{reg}} = -x_R - 0.1 u_t$. Default delay ($\Delta=6$ steps) and sharpness ($k=10$). Simulation horizon: 500 steps. Seeds: 50 per condition.

\subsection{Reproducibility}

All results are generated from frozen random seeds (seeds 1--50 for each condition). The simulation code is deterministic given a seed, configuration, and Python/NumPy version. Results reported in the main text are stored under \texttt{experiments/results/paper1\_v2/}. The regime classifier thresholds are documented in \texttt{experiments/src/autonomy\_lab/metrics.py} and were fixed prior to running the final experiments.

%% file: main.bbl
\begin{thebibliography}{35}
\providecommand{\natexlab}[1]{#1}
\providecommand{\url}[1]{\texttt{#1}}
\expandafter\ifx\csname urlstyle\endcsname\relax
  \providecommand{\doi}[1]{doi: #1}\else
  \providecommand{\doi}{doi: \begingroup \urlstyle{rm}\Url}\fi

\bibitem[Alboszta and Miekisz(2004)]{alboszta2004stability}
Jan Alboszta and Jacek Miekisz.
\newblock Stability of evolutionarily stable strategies in discrete replicator
  dynamics with time delay.
\newblock \emph{Journal of Theoretical Biology}, 231\penalty0 (2):\penalty0
  175--179, 2004.
\newblock \doi{10.1016/j.jtbi.2004.06.012}.

\bibitem[Ben-Khalifa et~al.(2018)Ben-Khalifa, El-Azouzi, and
  Hayel]{benkhalifa2018distributed}
Nesrine Ben-Khalifa, Rachid El-Azouzi, and Yezekael Hayel.
\newblock Discrete and continuous distributed delays in replicator dynamics.
\newblock \emph{Dynamic Games and Applications}, 8\penalty0 (4):\penalty0
  713--732, 2018.
\newblock \doi{10.1007/s13235-017-0225-7}.

\bibitem[Bouteiller et~al.(2021)Bouteiller, Ramstedt, Beltrame, Pal, and
  Binas]{bouteiller2020delays}
Yann Bouteiller, Simon Ramstedt, Giovanni Beltrame, Christopher Pal, and
  Jonathan Binas.
\newblock Reinforcement learning with random delays.
\newblock In \emph{International Conference on Learning Representations
  (ICLR)}, 2021.

\bibitem[Chen et~al.(2020)Chen, Xu, Liu, Li, and Zhao]{chen2020delay}
Baiming Chen, Mengdi Xu, Zuxin Liu, Liang Li, and Ding Zhao.
\newblock Delay-aware multi-agent reinforcement learning for cooperative and
  competitive environments.
\newblock \emph{arXiv preprint arXiv:2005.05441}, 2020.

\bibitem[Derman et~al.(2021)Derman, Dalal, and Mannor]{derman2021acting}
Esther Derman, Gal Dalal, and Shie Mannor.
\newblock Acting in delayed environments with non-stationary markov policies.
\newblock In \emph{International Conference on Learning Representations
  (ICLR)}, 2021.

\bibitem[Freeman(1977)]{freeman1977betweenness}
Linton~C. Freeman.
\newblock A set of measures of centrality based on betweenness.
\newblock \emph{Sociometry}, 40\penalty0 (1):\penalty0 35--41, 1977.
\newblock \doi{10.2307/3033543}.

\bibitem[Girvan and Newman(2002)]{girvan2002community}
Michelle Girvan and Mark E.~J. Newman.
\newblock Community structure in social and biological networks.
\newblock \emph{Proceedings of the National Academy of Sciences}, 99\penalty0
  (12):\penalty0 7821--7826, 2002.
\newblock \doi{10.1073/pnas.122653799}.

\bibitem[Gronauer and Diepold(2022)]{gronauer2022marl}
Sven Gronauer and Klaus Diepold.
\newblock Multi-agent deep reinforcement learning: A survey.
\newblock \emph{Artificial Intelligence Review}, 55:\penalty0 895--943, 2022.
\newblock \doi{10.1007/s10462-021-09996-w}.

\bibitem[Hale and Verduyn~Lunel(1993)]{hale1993introduction}
Jack~K Hale and Sjoerd~M Verduyn~Lunel.
\newblock \emph{Introduction to Functional Differential Equations}, volume~99
  of \emph{Applied Mathematical Sciences}.
\newblock Springer, 1993.

\bibitem[Hassard et~al.(1981)Hassard, Kazarinoff, and Wan]{hassard1981theory}
Brian~D Hassard, Nicholas~D Kazarinoff, and Yieh-Hei Wan.
\newblock \emph{Theory and Applications of Hopf Bifurcation}, volume~41 of
  \emph{London Mathematical Society Lecture Note Series}.
\newblock Cambridge University Press, 1981.

\bibitem[Hofbauer and Sigmund(1998)]{hofbauer1998evolutionary}
Josef Hofbauer and Karl Sigmund.
\newblock \emph{Evolutionary Games and Population Dynamics}.
\newblock Cambridge University Press, 1998.

\bibitem[Holland et~al.(1983)Holland, Laskey, and
  Leinhardt]{holland1983stochastic}
Paul~W. Holland, Kathryn~Blackmond Laskey, and Samuel Leinhardt.
\newblock Stochastic blockmodels: First steps.
\newblock \emph{Social Networks}, 5\penalty0 (2):\penalty0 109--137, 1983.
\newblock \doi{10.1016/0378-8733(83)90021-7}.

\bibitem[Iijima(2012)]{iijima2012delayed}
Ryota Iijima.
\newblock On delayed discrete evolutionary dynamics.
\newblock \emph{Journal of Theoretical Biology}, 300:\penalty0 1--6, 2012.
\newblock \doi{10.1016/j.jtbi.2012.01.001}.

\bibitem[Kuang(1993)]{kuang1993delay}
Yang Kuang.
\newblock \emph{Delay Differential Equations with Applications in Population
  Dynamics}.
\newblock Academic Press, 1993.

\bibitem[Lieberman et~al.(2005)Lieberman, Hauert, and
  Nowak]{lieberman2005evolutionary}
Erez Lieberman, Christoph Hauert, and Martin~A. Nowak.
\newblock Evolutionary dynamics on graphs.
\newblock \emph{Nature}, 433\penalty0 (7023):\penalty0 312--316, 2005.
\newblock \doi{10.1038/nature03204}.

\bibitem[McAvoy and Allen(2021)]{mcavoy2022social}
Alex McAvoy and Benjamin Allen.
\newblock Fixation probabilities in evolutionary dynamics under weak selection.
\newblock \emph{Journal of Mathematical Biology}, 82:\penalty0 14, 2021.
\newblock \doi{10.1007/s00285-021-01568-4}.

\bibitem[Miekisz(2008)]{miekisz2008evolutionary}
Jacek Miekisz.
\newblock Evolutionary game theory and population dynamics.
\newblock In \emph{Multiscale Problems in the Life Sciences}, volume 1940 of
  \emph{Lecture Notes in Mathematics}, pages 269--316. Springer, 2008.
\newblock \doi{10.1007/978-3-540-78362-6_5}.

\bibitem[Mittal et~al.(2020)Mittal, Mukhopadhyay, and
  Chakraborty]{mittal2020delayed}
Sourabh Mittal, Archan Mukhopadhyay, and Sagar Chakraborty.
\newblock Evolutionary dynamics of the delayed replicator--mutator equation:
  Limit cycle and cooperation.
\newblock \emph{Physical Review E}, 101\penalty0 (4):\penalty0 042410, 2020.
\newblock \doi{10.1103/PhysRevE.101.042410}.

\bibitem[Mohamadichamgavi and Bodnar(2025)]{gao2025bifurcation}
Javad Mohamadichamgavi and Marek Bodnar.
\newblock Bifurcation analysis of replicator dynamics with logistic growth and
  strategy-dependent time delays in snowdrift game.
\newblock \emph{Dynamic Games and Applications}, 2025.
\newblock \doi{10.1007/s13235-025-00671-1}.

\bibitem[Nowak(2006)]{nowak2006five}
Martin~A. Nowak.
\newblock Five rules for the evolution of cooperation.
\newblock \emph{Science}, 314\penalty0 (5805):\penalty0 1560--1563, 2006.
\newblock \doi{10.1126/science.1133755}.

\bibitem[Ohtsuki et~al.(2006)Ohtsuki, Hauert, Lieberman, and
  Nowak]{ohtsuki2006simple}
Hisashi Ohtsuki, Christoph Hauert, Erez Lieberman, and Martin~A. Nowak.
\newblock A simple rule for the evolution of cooperation on graphs and social
  networks.
\newblock \emph{Nature}, 441\penalty0 (7092):\penalty0 502--505, 2006.
\newblock \doi{10.1038/nature04605}.

\bibitem[Perc et~al.(2013)Perc, G{\'o}mez-Garde{\~n}es, Szolnoki, Flor{\'i}a,
  and Moreno]{perc2013evolutionary}
Matja{\v{z}} Perc, Jes{\'u}s G{\'o}mez-Garde{\~n}es, Attila Szolnoki, Luis~M.
  Flor{\'i}a, and Yamir Moreno.
\newblock Evolutionary dynamics of group interactions on structured
  populations: A review.
\newblock \emph{Journal of the Royal Society Interface}, 10\penalty0
  (80):\penalty0 20120997, 2013.
\newblock \doi{10.1098/rsif.2012.0997}.

\bibitem[Perc et~al.(2017)Perc, Jordan, Rand, Wang, Boccaletti, and
  Szolnoki]{perc2017statistical}
Matja{\v{z}} Perc, Jillian~J. Jordan, David~G. Rand, Zhen Wang, Stefano
  Boccaletti, and Attila Szolnoki.
\newblock Statistical physics of human cooperation.
\newblock \emph{Physics Reports}, 687:\penalty0 1--51, 2017.
\newblock \doi{10.1016/j.physrep.2017.05.004}.

\bibitem[Santos and Pacheco(2005)]{santos2005scalefree}
Francisco~C. Santos and Jorge~M. Pacheco.
\newblock Scale-free networks provide a unifying framework for the emergence of
  cooperation.
\newblock \emph{Physical Review Letters}, 95\penalty0 (9):\penalty0 098104,
  2005.
\newblock \doi{10.1103/PhysRevLett.95.098104}.

\bibitem[Scheffer(2009)]{scheffer2009critical}
Marten Scheffer.
\newblock \emph{Critical Transitions in Nature and Society}.
\newblock Princeton University Press, 2009.

\bibitem[Sutton and Barto(2018)]{sutton2018reinforcement}
Richard~S. Sutton and Andrew~G. Barto.
\newblock \emph{Reinforcement Learning: An Introduction}.
\newblock MIT Press, 2018.

\bibitem[Szab{\'o} and F{\'a}th(2007)]{szabo2007evolutionary}
Gy{\"o}rgy Szab{\'o} and G{\'a}bor F{\'a}th.
\newblock Evolutionary games on graphs.
\newblock \emph{Physics Reports}, 446\penalty0 (4--6):\penalty0 97--216, 2007.
\newblock \doi{10.1016/j.physrep.2007.04.004}.

\bibitem[Taylor and Jonker(1978)]{taylor1978evolutionary}
Peter~D. Taylor and Leo~B. Jonker.
\newblock Evolutionary stable strategies and game dynamics.
\newblock \emph{Mathematical Biosciences}, 40\penalty0 (1--2):\penalty0
  145--156, 1978.
\newblock \doi{10.1016/0025-5564(78)90077-9}.

\bibitem[Traulsen et~al.(2006)Traulsen, Nowak, and
  Pacheco]{traulsen2006stochastic}
Arne Traulsen, Martin~A Nowak, and Jorge~M Pacheco.
\newblock Stochastic dynamics of invasion and fixation.
\newblock \emph{Physical Review E}, 74\penalty0 (1):\penalty0 011909, 2006.
\newblock \doi{10.1103/PhysRevE.74.011909}.

\bibitem[Wesson and Rand(2016)]{wesson2016hopfRPS}
Elizabeth Wesson and Richard Rand.
\newblock Hopf bifurcations in delayed rock--paper--scissors replicator
  dynamics.
\newblock \emph{Dynamic Games and Applications}, 6\penalty0 (1):\penalty0
  139--156, 2016.
\newblock \doi{10.1007/s13235-015-0138-2}.

\bibitem[Wesson et~al.(2016)Wesson, Rand, and Rand]{wesson2016hopf}
Elizabeth Wesson, Richard~H. Rand, and David~G. Rand.
\newblock Hopf bifurcations in two-strategy delayed replicator dynamics.
\newblock \emph{International Journal of Bifurcation and Chaos}, 26\penalty0
  (1):\penalty0 1650006, 2016.
\newblock \doi{10.1142/S0218127416500061}.

\bibitem[Wettergren(2023)]{wettergren2023replicator}
Thomas~A. Wettergren.
\newblock Replicator dynamics of evolutionary games with different delays on
  costs and benefits.
\newblock \emph{Applied Mathematics and Computation}, 458:\penalty0 128228,
  2023.
\newblock \doi{10.1016/j.amc.2023.128228}.

\bibitem[Yan et~al.(2021)Yan, Chen, Qiu, and Szolnoki]{yan2021cooperator}
Fang Yan, Xiaojie Chen, Zhipeng Qiu, and Attila Szolnoki.
\newblock Cooperator driven oscillation in a time-delayed feedback-evolving
  game.
\newblock \emph{New Journal of Physics}, 23:\penalty0 053017, 2021.
\newblock \doi{10.1088/1367-2630/abf205}.

\bibitem[Zhang et~al.(2021)Zhang, Yang, and Ba{\c{s}}ar]{zhang2021marl_survey}
Kaiqing Zhang, Zhuoran Yang, and Tamer Ba{\c{s}}ar.
\newblock Multi-agent reinforcement learning: A selective overview of theories
  and algorithms.
\newblock In \emph{Handbook of Reinforcement Learning and Control}, Studies in
  Systems, Decision and Control, pages 321--384. Springer, 2021.

\bibitem[Zhang et~al.(2023)Zhang, Zhang, Gu, and Li]{zhang2023marl}
Yuyang Zhang, Runyu Zhang, Yuantao Gu, and Na~Li.
\newblock Multi-agent reinforcement learning with reward delays.
\newblock In \emph{Proceedings of the 5th Annual Learning for Dynamics and
  Control Conference (L4DC)}, volume 211 of \emph{PMLR}, pages 692--704, 2023.

\end{thebibliography}
